\documentclass[runningheads,final]{llncs}
\usepackage[T1]{fontenc}
\usepackage{graphicx}
\usepackage[T1]{fontenc}
\usepackage[utf8]{inputenc}
\usepackage{amsfonts}
\usepackage{amssymb}
\usepackage{amsmath}
\usepackage{mathtools}
\usepackage{thm-restate}
\usepackage{multirow}
\usepackage{float}
\usepackage{pgfplots
}
\usepackage{tablefootnote}
\usepackage{color}

\usepackage{cancel}
\usepackage{wasysym}
\usepackage{makeidx}
\usepackage{booktabs}
\usepackage{stmaryrd}
\usepackage{enumerate}
\usepackage{csquotes}
\usepackage{framed}
\usepackage{multido}
\usepackage[colorlinks=true,
            urlcolor=blue,
            linkcolor=black,
            citecolor=blue]{hyperref}
\usepackage{breakcites}
\usepackage{xargs}
\usepackage[most]{tcolorbox}
\usepackage{xspace}
\usepackage{tikz}
\usetikzlibrary{calc,positioning,arrows, backgrounds, math,shapes.symbols,shapes.geometric}
\usepackage{cleveref}

\newcommand{\bin}{\{0,1\}}

\newcommand{\N}{\mathbb{N}}                        
\newcommand{\R}{\mathbb{R}}                        
\newcommand{\Z}{\ensuremath{\mathbb{Z}}\xspace} 
\newcommand{\group}{\mathbb{G}}

\newcommand{\mat}[1]{\ensuremath{\mathbf{#1}}}

\newcommand{\ppt}{\ensuremath{\mathsf{PPT}}\xspace}
\newcommand{\PPT}{\ppt}

\newcommand{\cF}{\ensuremath{\mathcal{F}}\xspace}

\newcommand{\D}{\mathcal{D}}
\newcommand{\cR}{\mathcal{R}}
\newcommand{\cE}{\ensuremath{\mathcal{E}}\xspace}
\newcommand{\cS}{\mathcal{S}}

\newcommand{\thetab}{\boldsymbol{\theta}}

\newcommand{\PRF}{\ensuremath{\mathsf{PRF}}\xspace}
\newcommand{\RF}{\ensuremath{\mathsf{RF}}\xspace}

\newcommand{\Enc}{\ensuremath{\mathsf{Enc}}\xspace}
\newcommand{\Dec}{\ensuremath{\mathsf{Dec}}\xspace}
\newcommand{\Setup}{\ensuremath{\mathsf{Setup}}\xspace}
\newcommand{\KeyGen}{\ensuremath{\mathsf{KeyGen}}\xspace}

\newcommand{\EKeyGen}{\ensuremath{\mathsf{EKeyGen}}\xspace}

\newcommand{\poly}{\mathsf{poly}}

\newcommand{\mmax}{m_{\max}}
\newcommand{\nmax}{n_{\max}}

\newcommand\funcFam{\cF_S^{m_\ell}}
\newcommand{\resSpace}{\mathcal{Y}}

\newcommand{\msk}{\ensuremath{\mathsf{msk}}\xspace}

\newcommand{\ct}{\ensuremath{\mathsf{ct}}\xspace}
\newcommand{\ctil}{\ensuremath{\mathsf{ct}_{i,\ell}}\xspace}

\newcommand{\sk}{\ensuremath{\mathsf{ek}}\xspace}
\newcommand{\ski}{\ensuremath{\mathsf{ek}_i}\xspace}
\newcommand{\dk}{\ensuremath{\mathsf{dk}}\xspace}
\newcommand{\pp}{\ensuremath{\mathsf{pp}}\xspace}

\newcommand{\ep}{\ensuremath{T}\xspace}
\newcommand{\cM}{\ensuremath{\mathcal{M}}\xspace}

\newcommand{\sampler}{\ensuremath{\overset{\$}{\leftarrow}}\xspace}

\newcommand{\noise}{\ensuremath{\nu}\xspace}						 
\newcommand{\qk}{\ensuremath{q_{k}}\xspace}
\newcommand{\qkl}{\ensuremath{q_{k,\ell}}\xspace}
\newcommand{\ql}{\ensuremath{q_{\ell}}\xspace}
\newcommand{\qkkey}{\ensuremath{\qkl^{\KeyGen}}}
\newcommand{\qkckey}{\ensuremath{\qkl^{\mathsf{CKeyGen}}}}

\newcommand{\qcil}{\ensuremath{q_{c,i, \ell}}\xspace}
\newcommand{\qcjl}{\ensuremath{q_{c,j, \ell}}\xspace}
\newcommand{\qcenc}{\ensuremath{\qcil^{\Enc}}\xspace}
\newcommand{\qccenc}{\ensuremath{\qcil^{\mathsf{CEnc}}}\xspace}
\newcommand{\distr}{\ensuremath{\mathbb{D}}\xspace} 

\newcommand{\mife}{\ensuremath{\mathsf{MIFE}}\xspace}
\newcommand{\mcfe}{\ensuremath{\mathsf{MCFE}}\xspace}

\newcommand{\nmcfe}{\ensuremath{\mathsf{NMCFE}}\xspace}
\newcommand{\dynmcfe}{\ensuremath{\mathsf{DyNMCFE}}\xspace}

\newcommand{\lnmcfeot}{\ensuremath{\mathsf{DyNo}}\xspace}
\newcommand{\lnmcfe}{\ensuremath{\lnmcfeot}\xspace}
\newcommand{\nmifeot}{\ensuremath{\mathsf{NMIFE}^{\textnormal{ot}}}\xspace}
\newcommand{\mifeot}{\ensuremath{\mathsf{MIFE}^{\textnormal{ot}}}\xspace}

\newcommand{\yb}{\ensuremath{\mathbf{y}}\xspace}
\newcommand{\xb}{\ensuremath{\mathbf{x}}\xspace}

\newcommand{\HS}{\ensuremath{\mathcal{HS}}\xspace}

\newcommand{\CS}{\ensuremath{\mathcal{CS}}\xspace}

\newcommand{\Labels}{\ensuremath{\mathcal{L}}\xspace}

\newcommand\GM{\mathsf{GM}}

\newcommand{\QEnc}{\ensuremath{\mathsf{QEnc}}\xspace}
\newcommand{\QEKeyGen}{\ensuremath{\mathsf{QEKeyGen}}\xspace}
\newcommand{\CQEnc}{\ensuremath{\mathsf{CQEnc}}\xspace}
\newcommand{\QCor}{\ensuremath{\mathsf{QCor}}\xspace}
\newcommand{\QKeyGen}{\ensuremath{\mathsf{QKeyGen}}\xspace}

\newcommand{\CQKeyGen}{\ensuremath{\mathsf{CQKeyGen}}\xspace}

\newcommand{\game}[1]{\ensuremath{\mathsf{G}_{#1}}\xspace}
\newcommand{\Win}[1]{\ensuremath{\mathsf{P}(\adv, \game{#1})}\xspace}

\newcommand{\advantage}{\ensuremath{\mathsf{Adv}}\xspace}
\newcommand{\adv}{\ensuremath{\mathcal{A}}\xspace}					 
\newcommand{\cB}{\ensuremath{\mathcal{B}}\xspace}
\newcommand{\negl}{\ensuremath{\mathtt{negl}}\xspace}
\newcommand{\challenger}{\ensuremath{\mathcal{C}}\xspace}

\newcommand{\ww}{\textnormal{ww}}
\newcommand{\xx}{\textnormal{xx}}
\newcommand{\yy}{\textnormal{yy}}
\newcommand{\zz}{\textnormal{zz}}
\newcommand{\sta}{\textnormal{sta}}
\newcommand{\ad}{\textnormal{ad}}
\newcommand{\mh}{\textnormal{mh}}
\newcommand{\nh}{\textnormal{nh}}

\newcommand{\anyot}{\ensuremath{\textnormal{any}^{\textnormal{ot}}}}
\newcommand{\pos}{\textnormal{one}}
\newcommand{\posp}{\textnormal{pos}}
\newcommand{\lab}{\textnormal{lab}}
\newcommand{\all}{\textnormal{all}}
\newcommand{\sel}{\textnormal{sel}}
\newcommand{\IND}{\ensuremath{\mathsf{IND}}\xspace}

\renewcommand{\cref}{\Cref}

\ifdefined\eprint
\AtEndEnvironment{proof}{\phantom{}\qed}
\fi

\makeatletter
\newcommand*{\centerfloat}{%
  \parindent \z@
  \leftskip \z@ \@plus 1fil \@minus \textwidth
  \rightskip\leftskip
  \parfillskip \z@skip}
\makeatother

\newcommand{\graybox}[1]{\colorbox{gray!48}{#1}}
\newcommand{\gbox}[1]{\colorbox{gray!18}{#1}}
\newcommand{\grayindex}[1]{\colorbox{gray!48}{$\scriptstyle #1$}}
\newcommand{\gindex}[1]{\colorbox{gray!18}{$\scriptstyle #1$}}

{\endMakeFramed}

\definecolor{jas-color}{cmyk}{0.0, 0.72, 0.72, 0.28}
\definecolor{notes-color}{cmyk}{.06, 0, 1.0, .5}
\definecolor{linda-color}{cmyk}{0.8, 0.2, 0.5, 0.0}
\definecolor{erkan-color}{cmyk}{0.0, 0.42, 1.0, 0.16}

\definecolor{cb1}{rgb}{0.0, 0.0, 0.502}
\definecolor{cb2}{rgb}{1.0, 0.502, 0.0}
\definecolor{cb3}{rgb}{0.0, 0.502, 0.0}
\definecolor{cb4}{rgb}{0.663, 0.663, 0.663}

\newcounter{construction} 
\renewcommand{\theconstruction}{\arabic{construction}}

\newenvironment{construction}[1][]{%
  \refstepcounter{construction}%
  \par\medskip
  \noindent{\bfseries Construction~\theconstruction.%
  \ifx&#1&\else~(#1)\fi}%
  \itshape
}{\par\medskip}

\begin{document}

\title{Enhancing Noisy Functional Encryption for Privacy-Preserving Machine Learning}
\titlerunning{Enhancing Noisy FE for PPML}
\author{Linda Scheu-Hachtel \and
Jasmin Zalonis
}
\authorrunning{L. Scheu-Hachtel and J. Zalonis}
%
\institute{Univerity of Mannheim, Mannheim, Germany \\
\email{\{linda.scheu-hachtel, zalonis\}@uni-mannheim.de}}

\maketitle

\begin{abstract}
	Functional encryption (FE) has recently attracted interest in privacy-preserving machine learning (PPML) for its unique ability to compute specific functions on encrypted data.
	A related line of work focuses on noisy FE, which ensures differential privacy in the output while keeping the data encrypted.
		We extend the notion of  noisy multi-input functional encryption (NMIFE) to (dynamic) noisy multi-client functional encryption ((Dy)NMCFE), which allows for more flexibility in the number of data holders and analyses, while protecting the privacy of the data holder with fine-grained access through the usage of labels.

		Following our new definition of DyNMCFE, we present $\mathsf{DyNo}$, a concrete inner-product DyNMCFE scheme.
		Our scheme captures all the  functionalities previously introduced in noisy FE schemes, while being significantly more efficient in terms of space and runtime and fulfilling a stronger security notion by allowing the corruption of clients.

		To further prove the applicability of DyNMCFE, we present a protocol for PPML based on $\mathsf{DyNo}$.
		According to this protocol, we train a privacy-preserving logistic regression.

\end{abstract}


\section{Introduction}\label{sec:intro}
Privacy-preserving machine learning (PPML) is crucial for extracting insights from large datasets in sensitive domains like finance, healthcare, and social research.
The necessity to protect the privacy of the data holders sparked an extensive research in PPML.

In this work, we consider a medical-inspired scenario, applicable to other contexts as well.
We assume the presence of a trusted authority, e.g., an ethics committee, that provides a platform for PPML, for example as a sub-feature of an eHR app \cite{eHR}.
On this platform, data holders with sensitive information can voluntarily register to participate in analyses.
Analysts can also register, initiate new studies, and invite data holders to contribute their data securely.

As an example, assume an analyst wants to train a model to predict risk factors for lung cancer.
They can open an analysis with requested attributes "age", "smoker", "BMI", "chronic disease" and "lung cancer".
Then, the data holders can submit their personal data for these requested attributes to the analyst, who stores the data under some pre-determined label, or decline/ignore the request.
Once enough data is provided, the analyst and other interested parties in the same attributes can run their analysis. 
However, we assume all of the data holders will only participate if:
\begin{itemize}
	\item[i)] They only have to provide their data and not actively engage in the analysis.
	\item[ii)] Their privacy remains protected.
\end{itemize}
Additionally, we assume that the authority has limited resources and should not receive the plaintexts or learn the result of the analysis.
To address ii), we must not only protect the plain data, ensuring \emph{input} privacy, but also require additional measures to safeguard \emph{output} privacy, as the analysis results may leak information \cite{carpov2020illuminating}.

The only common approach to obtain both directly is local differential privacy (LDP)~\cite{ldp}.
LDP perturbs data at its source, either through randomized response or addition of noise, such that the original values cannot be determined.
This way, the perturbed data can be collected on a public sever, where any analyst can use it without further user involvement.
However, this approach requires an enormous amount of perturbation, which severely impacts  the utility of any model trained on such data \cite{ploss1, ploss2}.

Therefore, a lot of research focuses on techniques combining mechanisms providing input privacy with global differential privacy (GDP) \cite{Dwork14}.
Unlike LDP, GDP perturbs the analysis output rather than the raw data, resulting in lower utility loss.
A widely used GDP-based approach for PPML is federated learning (FL) \cite{li2020federated}, but since the raw data remains untouched, additional protection is required when data is outsourced.
Other common approaches for PPML combining input privacy with GDP are based on multi-party computation (MPC) \cite{goldreich1998secure, keller2020mp, zhao2019secure} and (fully) homomorphic encryption ((F)HE)~\cite{gentry2009fully,kim2018logistic,kim2019secure},
each offering different trade-offs discussed in \cref{sec:approaches}.

Recently,  functional encryption (FE)~\cite{panzade2023fenet,xu2019cryptonn,chang2023privacy}, has been widely discussed to be helpful in PPML.
FE overcomes the all-or-nothing decryption limitations of standard encryption schemes, as it provides designated decryption keys.
A decryption key $\dk_f$ is associated with a specific function $f$  and instead of decrypting directly to the plain input, the decryption key applied to the ciphertext only reveals the decryption of $f(x)$ and nothing more.
Further, it allows to outsource some of the computational overhead to the analyst, which is necessary if the trusted party has limited resources and the data holders are not actively involved, i.e. i) is fulfilled.

Zalonis et al.~\cite{ZASH24} introduced the notion of noisy multi-input functional encryption (NMIFE), which is a method to combine GDP and multi-input functional encryption (MIFE).
The idea is to equip the functional decryption key $\dk_f$ with a noise value $\noise$ sampled according to some distribution to obtain $f(x_1, \dots, x_n)+ \noise$ instead of $f(x_1, \dots, x_n)$.
By choosing an appropriate distribution, this is enough to realize GDP.
The security of NMIFE schemes ensures that the decryption key does not leak anything about the noise value.

Unfortunately, while a lot of research in FE aims to realize arbitrary functions, i.e., general circuits, these techniques require heavy tools such as indistinguishability obfuscation or polynomial hardness of assumptions on multilinear maps~\cite{TCC:BoyChuPas14, garg2016candidate, waters2015punctured}.
As a result, these schemes are far from being practical.
Currently, efficient (N)MIFE schemes are only feasible for specific families of functions, namely linear functions~\cite{C:AgrLibSte16, PKC:DatOkaTom18, PKC:MKMS22} and quadratic functions~\cite{baltico2017practical, agrawal2022multi}.
However, all of the latter require pairings on bilinear groups, making them rather slow.
Thus, we focus on linear functions, which can already be used for many applications.
Specifically, we present a protocol leveraging ideas from FL, which allows the training of, e.g., linear or logistic regression models.

To meet the requirements of the use case presented above,
we extend the definition of NMIFE to noisy multi-client functional encryption (NMCFE), similar as from MIFE to multi-client functional encryption (MCFE)~\cite{AC:CDGPP18}, to achieve a stronger security notion and inherit additional benefits relevant to our scenario.
Unlike NMIFE, NMCFE equips each ciphertext with an additional label. 
Only those ciphertexts encrypted under the same label can be combined.
This enables a more fine-grained access as it prevents the combination of ciphertexts with different labels and thus mitigates mix-and-match attacks.
In our scenario, these labels could refer to the name of the study or dataset, but could also be time stamps for a long-term medical study.
Additionally, in contrast to NMIFE, the security of NMCFE allows clients to be corrupted, i.e., the adversary may possess the encryption keys of corrupted clients without compromising the security and privacy of the remaining clients.

We further extend NMCFE to a dynamic setting, where  data holders can dynamically register to be part of the dynamic NMCFE (DyNMCFE) for future analysis or drop out without the need to set up a new scheme any time the participants change.
As all analyses remain within the same scheme and are overseen by one trusted party, e.g. an ethical committee, this party can easily keep the privacy loss of each participating data holder  under control.

With DyNMCFE, we obtain a perfectly tailored primitive to our described scenario.
We present a concrete, efficient scheme that supports linear functions.
Like all other (MI)FE schemes supporting GDP, it only allows one encryption per client.
By using labels, we allow one encryption per client-label pair, effectively enabling multiple encryptions per client across different labels.
This may still seem like a limitation, however when training ML models using a predefined dataset is common, especially in a medical context and when combining ML with GDP.
Therefore if we assume that each client submits its data with specific attributes only once to a dataset associated with a given label, allowing multiple encryption per slot is unnecessary.

\paragraph{Contributions}
Our contributions are the following.
\begin{itemize}
	\item We formalize the notion of DyNMCFE and its security.
	This is a natural combination from regular MCFE and the definition of NMIFE provided in~\cite{ZASH24} enhanced to a dynamic setting. 
	In particular, we achieve a stronger security notion as NMIFE, as we allow the dynamic registration of clients and corruptions.
	\item We provide a protocol capturing the described scenario above using DyNMCFE for linear functions as our main cryptographic building block. This protocol ensures both input and output privacy and can be used to realize efficient PPML.
	\item Following our new definition and requirements for the protocol, we present \lnmcfeot, an efficient DyNMCFE scheme for inner-product functionality, allowing for a polynomial number of clients.
	Its security solely relies on a pseudorandom function (PRF).
	This construction is sufficient to provide GDP and more efficient than any previous noisy FE scheme for inner-product functionality.
	Tested on datasets of different sizes, our implementation operates in milliseconds, whereas concurrent schemes take seconds, if not minutes or hours.
	\item
	In addition, we use our protocol to train a differentially private logistic regression on encrypted data, which satisfies input as well as output privacy.
	This is the first logistic regression realized by an inner-product FE scheme.
\end{itemize}

\paragraph{Outline}
The paper is structured  as follows.
In Section~\ref{sec:prelim}, we present the necessary background on DP and the private gradient descent (GD) algorithm.
\cref{sec:approaches} gives an overview and comparison of different approaches to realize PPML  and their suitability for the specific scenario described above.
The class of DyNMCFE is introduced in Section~\ref{sec:nmcfe}.
We propose the protocol for our application in Section~\ref{sec:ppfl}.
Based on the requirements, we present our concrete linear DyNMCFE construction in Section~\ref{sec:noisy-construct}.
Section~\ref{sec:implementation} contains the benchmarking of our scheme as well as the results of the trained logistic regression.
We conclude this work in Section~\ref{sec:conclusion}.

\subsection{Related Work}
\paragraph{Noisy Functional Encryption}
Bakas et al.~\cite{bakas2022heal,bakas2022private} were the first to combine FE and DP.
They used the construction of \cite{C:ACFGU18} for inner-product MIFE to obtain GDP queries to an encrypted database.
However, they only considered the aggregated sum and not weighted sums.
Moreover, they lack a proper security definition for incorporating  noise in the decryption key.

	Zalonis et al.~\cite{ZASH24} filled this gap by introducing noisy FE together with the notion of noise-hiding security and a tailored correctness definition.
	In particular, noise-hiding allows for the function to be known by the decryptor, but it demands the noise to be hidden.
	This relaxation from function-hiding, where the decryption key leaks nothing about the function itself, allowed them to transform a function-hiding into a more efficient noise-hiding scheme.
	This scheme allows only one message to be encrypted per client, which is already sufficient for their discussed use case of PPML.
	However, their scheme is still based on pairings, which are known to be slow.
	Compared to our construction, their approach is less efficient in both size and runtime and neither supports the corruption of clients nor the dynamic generation of encryption keys.

	In an independent work, Escobar et al.~\cite{escobar2024computational} also combined GDP and FE.
	They consider an encrypted dataset to which an adversary may pose queries via an FE scheme.
	Their security definition is coined on computational DP, i.e., an adversary cannot efficiently distinguish between the outputs of the queries on neighboring datasets.
	Our security definition, extended from \cite{ZASH24}, is much broader, as it only considers GDP as a use case.
	Moreover, we consider the multi-client setting, whereas they only have a single data holder encrypting the whole dataset.
	Note that this is also covered by our definition if we set the number of clients to be one.


	A comparison between our scheme and the schemes from Zalonis et al. and Escobar et al. in size of the master secret key $\msk$, the encrypted database $ X$ with $n$ plaintext of size $m$ and decryption key $\dk$ can be found in Table~\ref{tab:size}.
	Here, $\group, \group_1$ and $\group_2$ are cyclic groups of order $q$ and $|\cdot|$ denotes the size of their elements.
	We ignore public information such as the function embedded inside the decryption key, the associated label and client, as those are negligible and part of any protocol.
	Also, we stress that all of these schemes only allow the encryption of one ciphertext per data holder.

	\begin{table}
		\caption{Comparison of sizes of all current FE schemes for data in $\Z_q$ realizing GDP, where $n$ is the number of data holders, $m$ is the vector length of the data and $\lambda$ the security parameter.}\label{tab:size}
		\centering
		\scalebox{1.}{
			\begin{tabular}{|c|c|c|c|}
			\hline
			& $\boldsymbol{|\msk|}$ & $\boldsymbol{|\Enc(X)|}$ & $\boldsymbol{|\dk|}$\\
			\hline
			\cite{ZASH24} & $n(m+3)|\group_2|$ & $n(m+5)|\group_1|$ & $(m+5) |\group_2|$\\
			\cite{escobar2024computational} & $3 \lambda$ & $(nm +2)|\group|$ & $4 \log q$\\
			\text{Ours} & $n\lambda$\tablefootnote[1]{This can be reduced to $\lambda$ if the clients' keys are generated using, e.g., a PRF. However, this introduces an additional overhead in other algorithms and only makes sense in the non-dynamic version.} & $nm \log q$ & $\log q$\\
			\hline
		\end{tabular}}

	\end{table}


\paragraph{Privacy-Preserving Machine Learning using Functional Encryption}
	There exist some works which combine ML and FE.
	Although most of them simply focus on prediction~\cite{carpov2020illuminating,dufour2018reading,ligier2017privacy,ryffel2019partially}, some also include training.
	For instance, Xu et al.~\cite{xu2019cryptonn} employed a 5-layer neural network by observing that the input to the activation function is simply an inner-product of the weights and the data.
	This inner-product in the first layer is computed using their own presented scheme.
	All of the other computations are made in the clear.
	Panzade et al.~\cite{panzade2023fenet} improved the speed of the scheme using another (function-hiding) IPFE.
	However, both of them leak the intermediate results and are only selective-IND-CPA secure, whereas our scheme is adaptively secure and provides output privacy.

	Recently, Chang et al.~\cite{chang2023privacy} combined FL and FE using several MCFE schemes.
	Their model updates are merged using MCFE, which concludes a lower privacy loss, but still none that provides DP.
	Moreover, they need to re-encrypt their ciphertexts in each iteration, whereas our data holders only submit their data once and do not have to participate in the analysis.

\section{Preliminaries}\label{sec:prelim}
\subsection{Notation}
We always denote the security parameter by $\lambda \in \N$,
which parametrizes each scheme and adversary.
For integers $m, n \in \Z$, let $[m;n] := \{x \in \Z| m \leq x \leq n\}$ and $[n] := [1;n]$.  For $x \in \R$, $\lfloor x \rceil$ is its rounding to the nearest integer.
For clarity, we use bold symbols to denote vectors, i.e., $\xb$ is a vector, where $\xb[j]$ is its $j^{th}$ element.
In contrast, objects of some collection that is not regarded as a vector are indexed using subscripts (or superscripts in some cases).
For example, $\xb_i$ represents a vector, not a component of some vector.
If $i$ runs through some index set $[n]$, it means that there are $n$ vectors $\xb_1,\dots, \xb_n$.
If the $n$ objects are scalars (or not explicitly vectors), we will write $x_1, \dots, x_n$ instead.
For a function $f: \mathcal{X}_1 \times \cdots \times \mathcal{X}_n \rightarrow \mathcal{Y}$ and some $\noise\in \mathcal{Y}$, we denote by  $(f+\noise)(x) := f(x) + \noise$ for all $x \in \mathcal{X}_1 \times \cdots \times \mathcal{X}_n$.
For a subset $\Delta\subseteq Y$, we define $f+\Delta =\{f+\noise\mid \noise\in\Delta\}$.
For a set $S \subseteq \N$, we sometimes write $f(\{x_i\}_{i \in S}) = f(x_{i_1}, \dots, x_{i_{|S|}})$ for $i_1 \leq \dots \leq i_{|S|}$, $i_j \in S,~ j \in [|S|]$.

We use the abbreviation $\ppt$ to mean probabilistic polynomial time.
A function $\negl: \mathbb{N} \rightarrow \mathbb{R}^+$ is said to be negligible if for every $c \in \mathbb{N}$, there exists a $\tau \in \mathbb{N}$ such that for all $x \in \mathbb{N}$ with $x > \tau$, it holds that  $|\negl(x)| < 1/x^c$.

For a distribution \distr over some set $\Delta$, $\Pr[x \leftarrow \distr(\Delta) ]$ denotes the probability that $x$ is sampled according to  $\distr$ over $\Delta$.
If $\distr(\Delta)$ is clear from the context, we simply write $\distr$.
For any set $\Delta$, $s \sampler \Delta$ represents the process of uniformly sampling an element $s \in \Delta$.
For two probability distributions $\distr$ and $\distr'$ over the same domain $\Delta$, we write $\distr \equiv \distr'$ to indicate that they are equally distributed.

We restrict ourselves to  inner-product MIFE schemes.
More precisely, we consider functions of the form
	$f_{\yb_1, \dots, \yb_n} (x_1, \dots, x_n) = \sum_{i \in [n]} \langle \xb_i, \yb_i \rangle,$
i.e., each function can be identified with a vector \(\yb = (\yb_1, \dots, \yb_n) \in (\Z^{m})^{n}\).
The class of multi-input inner-product functionalities over \(\Z_q\) is defined by
$\cF_{q,n}^m := \{	f_{\yb_1, \dots, \yb_n}: (\Z_q^m)^n \to \Z_q \mid \yb_i \in \Z_q^m \}.$
Whenever we write $f_{\yb} \in\cF_{q,n}^m $, we assume that $\yb$ can be divided into $n$ equal vectors of size $m$.

\subsection{Differential Privacy}
DP ensures that changing or removing a single individual from a private dataset does not alter the output of an evaluation over this dataset by more than a negligible margin.
The overall result should remain approximately the same regardless of whether an individual participates in the evaluation or not.
Hence, their private attributes remain protected, but it is still possible to infer some information from the dataset.
One important property of DP is its robustness under post-processing~\cite{Dwork14}.
In other words, if the output of a DP mechanism is further processed by another function, the result still fulfills DP.

There are two main forms of DP, LDP and GDP. Since the latter is more common and assumed as the classical DP, \textit{global} is often omitted when it is clear from the context.
Formally, we define LDP as follows:
%
\begin{definition}[Local Differential Privacy~\cite{ldp}]\label{def:ldp}
	A randomized mechanism $f: \mathcal{X} \to \mathcal{Y}$ is $\epsilon$-local differentially private if for all $S \subseteq$ Range($f$) and for all $x, x' \in \mathcal{X}$:
	$$ \Pr[f(x) \in S] \leq \exp(\epsilon) \Pr[f(x') \in S], $$
	where the probability space is over the coin flips of the mechanism $f$.
\end{definition}

In LDP each record is perturbed directly \emph{before} the analysis, whereas in GDP it is perturbed \emph{during} or \emph{after} the analysis.
We define GDP as follows.
Denote by $X \simeq X'$ that two databases $X, X' \in \mathcal{X}$ are adjacent, i.e., differ in only one entry.
\begin{definition}[Global Differential Privacy]
	A randomized mechanism $f: \mathcal{X} \to \mathcal{Y}$ is ($\epsilon, \delta$)-differentially private if for all $S \subseteq$ Range($f$) and for all adjacent inputs $X\simeq X'$:
	$$ \Pr[f(X) \in S] \leq \exp(\epsilon) \Pr[f(X') \in S] + \delta, $$
	where the probability space is over the coin flips of the mechanism $f$. If $\delta = 0$, we say that $f$ is $\epsilon$-differential private.
\end{definition}
%

%
Typically, in the case of GDP, there is a curator, which has hold of the data $X$, and an analyst, who wants to evaluate some function $f$.
The curator returns a perturbed evaluation of the function, i.e., $\widetilde{f}(X) = f(X) + \noise$.
This noise $\noise$ is sampled from a distribution whose parameters are dependent on the leakage of $f$ and the so called privacy budget of the individuals.
As the analyst might query more than one function, it is important that the authority keeps track of the privacy budget and does not answer any more queries once it is consumed.
More precisely, if the analyst queries $n$ functions $f_i$, $i \in [n]$ which are $(\epsilon_i, \delta_i)$ private, then the consumed privacy budget is $(\sum_{i \in [n]} \epsilon_i, \sum_{i \in [n]} \delta_i)$.

In order to set the parameters of output perturbing distributions, we make use of the $l_2$-sensitivity.
This is a measure of how much one individual, i.e., one record, influences the output of a function.
\begin{definition}[$l_2$-sensitivity]
	The $l_2$-sensitivity of a function $f: \mathcal{X} \to \mathcal{Y}$ is
	$ \Delta_2(f) = \max_{X,X' \in \mathcal{X}, X \simeq X'} ||f(X) - f(X')||_2.$
\end{definition}
%

Balle et al.~\cite{balle18a} introduced the so called analytic Gaussian mechanism, that is mostly used in the ML context.
\begin{definition}\label{def:anagm}
	Let $f: \mathcal{X} \to \R^d$ be a function and $\Delta_2(f)$ its sensitivity.
	Let $Z$ be a $d$-dimensional, centered, independent Gaussian random variable with variance $\sigma^2$ and $u= \frac{\Delta_2(f)}{\sigma}$.
	For any $\epsilon \geq 0$ and $\delta \in [0,1]$, the analytic Gaussian output perturbation mechanism $\widetilde{f}(x) = f(x) +Z$
	is $(\epsilon, \delta)$-DP if and only if
	\begin{align}\label{eq:gm}
		\Phi\left(\frac{u}{2} - \frac{\epsilon }{u}\right) - \exp(\epsilon) \Phi\left(-\frac{u}{2} - \frac{\epsilon}{u}\right) \leq \delta,
	\end{align}
	where $\Phi$ denotes the Gaussian cumulative distribution function.
\end{definition}

Let $\GM(\epsilon, \delta, \Delta_2(f)) \to \distr$ be the algorithm that on input of the privacy parameters and the sensitivity outputs a $d$-dimensional, centered, independent Gaussian distribution whose parameter fulfill~(\ref{eq:gm}).
$\GM$ can be implemented efficiently using binary search.

\subsection{Private Gradient Descent}
In order to train a ML model, or more precisely its parameters $\thetab$, one common approach is to use GD~\cite{abadi2016deep}, which is an iteration based optimization algorithm.
The goal is to minimize the loss function $L(\thetab, \cdot)$ of the model.
In private GD, additional privacy measures are applied such as the addition of DP noise.

Let $X= \{(\xb_i)\}_{i \in [n]}$ be a dataset with $m$ attributes, where $\xb[0] := 1$ is a constant, $(\xb_i[1], \dots, \xb_i[m])$ the record and $\xb[m+1] \in \{0,1\}$ the dependent variable.
Moreover, let $\alpha$ be the learning rate and $\nabla L$ the gradient of $L$.
Let $F_{\thetab}(X) := \frac{\alpha}{n} \sum_{i \in [n]} \nabla L(\thetab, (\xb_i))$.
In traditional GD, the parameters are updated in each iteration $t$ by
\begin{align}\label{eq:privgrad}
    \thetab^{(t+1)} := \thetab^{(t)}+ F_{\thetab}^{(t)}(X).
\end{align}
In contrast, in private GD, the parameters are adjusted by applying additional privacy-preserving noise, i.e.,
\begin{align}\label{eq:dpprivgrad}
    \thetab^{(t+1)} := \thetab^{(t)} + F_{\thetab}^{(t)}(X) + \boldsymbol\noise^{(t)},
\end{align}
where $\boldsymbol\noise^{(t)}$ is sampled according to a DP distribution.

\section{Approaches to Privacy-Preserving Machine Learning}\label{sec:approaches}

The advantages and disadvantages of using MPC, HE and FE in the context of PPML have been studied in detail in, e.g, \cite{ZASH24,chang2023privacy}.
In this section, we expand their comparison to our specific scenario and focus on the applicability of solutions based on LDP, FL, MPC, HE and FE.
Recall that our scenario requires that:
\begin{itemize}
	\item[i)]  The trusted party has limited resources.
	\item[ii)] Data holders can join dynamically.
	\item[iii)] Data holders do not have to actively participate during the analysis.
	\item[iv)] The input and output privacy is ensured.
	\item[v)] The analyst decides on when the analysis is performed.
\end{itemize}

For MPC and HE we can always obtain DP through (one of) the trusted parties, thus we can assume that output privacy is always given in the scenarios.
\paragraph{Local Differential Privacy}
In LDP, each data holder locally randomizes their data using some mechanism $\cM$ such that \Cref{def:ldp} is satisfied.
After that the data is collected and can be stored on a publicly available server, without compromising the privacy.
Through the randomization, the data does not leak the plaintexts.
A major advantage of LDP is that individuals perturb their data once and do not have to participate any further, independent on the analysis.
By the post-processing property of (L)DP, the privacy is protected, no matter how many functions are evaluated on the data.
Moreover, there is no need for a trusted party, especially one with demanding resources.

However, while LDP provides stronger individual privacy guarantees and fulfills all of the above requirements, this comes at the cost of severely degraded data utility due to the high level of noise that needs to be added to each individual data point \cite{ploss1, ploss2}.

\paragraph{Federated Learning}
FL is a well known approach to enable the training of ML models via an iterative GD algorithm over partitioned data, without revealing the plain or pseu\-do\-nym\-ized data.
Usually, FL is used if data holder possess multiple data records.
It has been getting an increasing interest not only from academia but also from industry \cite{tensorflow2015, ziller2021pysyft, fan2023fate}.

The core idea is that each data holder trains a local model on their local data and that these models are merged.
More precisely, each data holder trains some iterations of the GD on their local data, optimizing their weights.
Then, the trained weights are aggregated over all parties and again distributed to resume the local training with aggregated weights.

FL was proposed to preserve privacy during training by sharing no plain data but only the intermediate results after the iterations.
As these may still leak too much information, before sharing the local intermediate results, some well chosen noise is added to conceal the concrete result and provide DP.

With FL we can obtain i) and iii)- v).
However, as the model is trained locally, this requires the active participation of the data holders during training, which violates iii). 
Moreover, if we assume a minimal trust model, i.e., as few data holders as possible, each data holder only possesses one record. Hence, the overall amount of noise per local model is significant, almost resembling LDP.

\paragraph{Multi-Party Computation}
In MPC, a given function is jointly computed by a set of parties, where each parties obtains a so-called share of the input data from the data holders.
To obtain, iv), it is important that no party is in possession of all shares.
Otherwise, decryption is possible and the plain data is leaked.

While MPC protocols are highly efficient, and the data holders do not need to be involved after providing their data, fulfilling iii), it still has some disadvantages.
First of all, it is important that no untrusted party is in possession of all of the shares, as otherwise decryption would be possible.
Hence, the security of most MPC protocols either requires the majority of the computing parties to be trusted, or only considers a passive adversary who honestly follows the protocol.
If the data is collected over a period of time, as considered in our dynamic use case, each of the computing parties has to potentially store the received shares over a long time.
They cannot be collected by the analyst as they would otherwise be able to decrypt.
Lastly, there are high communication costs between the parties. Thus, a communication network with strong delivery guarantee is required.
In conclusion, it is hard to fulfill conditions i), ii) and v).

\paragraph{Homomorphic Encryption}
The idea of HE is that one can operate over encrypted data similar than over plain data, i.e. for certain functions $f$ one can also compute $f^*$ over the ciphertexts, where $\mathtt{Dec}(\sk, f^*(\ct_{1}, \dots, \ct_{n}))$ yields $f(\xb_1, \dots, \xb_n)$. 
In particular, HE allows for complex functions without revealing any intermediate result, ensuring iv).
However, this comes at the cost of high computational power of the executing party.

As the key used for decryption is independent of the computation applied to the ciphertexts, it should not be in possession of the party that holds the ciphertexts.
This leaves two options:
First, the analyst is in possession of the decryption key, which implies that we are in need of a trusted computing party with the necessary resources to evaluate the function on the ciphertexts.
The analyst only obtains the final ciphertext, which contains the (possibly perturbed) function evaluation.
Second, the authority is in possession of the decryption key and the analyst performs the computation.
In this case, additional measures are necessary to verify that the analyst has correctly computed the intended function, e.g., using verifiable HE \cite{viand2023verifiablefullyhomomorphicencryption}.
However, verifications either require a trusted execution environment or zero-knowledge proofs, which yield an excessive overhead in the complexity of the function.
Hence, although ii) and iii) are fulfilled, in both scenarios the trusted party is in need of a lot of resources, contradicting i).

\paragraph{Noisy (Dynamic Multi-Client) Functional Encryption}
Classical noisy FE \cite{ZASH24} allows decryption of a noisy function's evaluation on encrypted data, i.e., any party in possession of a decryption key for a function $f$ can obtain $f(x_1, \dots, x_n) + \noise$ from the ciphertexts and decryption key and nothing more, ensuring iv).
Therefore data holders only need to participate by submitting their encrypted data to the analyst, giving v), while the rest of the communication is between the analyst and authority, providing iii).
Moreover, the generation of decryption keys is comparably low effort, fulfilling i).

However, to obtain ii), we have to resort to our new definition of DyNMCFE, which also covers all of the advantages above.
Hence, DyNMCFE seems like the perfect fit for our scenario.
More precisely, the data holders encrypt their data under the label $\ell$ using an encryption key from the trusted authority, which they receive upon registration.
After they submit their ciphertexts to the analyst, they are not a part of the analysis anymore and only act if they want to join a different analysis.
All further communication will then be between the authority and the analyst, who queries functions to receive decryption keys for their analysis.

The only disadvantage is the limitation in the function's complexity.
We solve this drawback in \cref{sec:ppfl} and provide a protocol how ML models such as linear or logistic regression can be trained with only a linear scheme.

\paragraph{Conclusion} The best fits for the described scenario, where conditions i)-v) are fulfilled, are DyNMCFE and LDP.
All that is left is the utility of the trained models of the analyst.
For this reason, we use LDP as a baseline to compare our utility to in \cref{sec:implementation}.
\section{Noisy Multi-Client Functional Encryption}
\label{sec:nmcfe}
In the following we combine the established notions of MCFE \cite{AC:LibTit19} and NMIFE \cite{ZASH24} and extend it to the new class of (Dy)NMCFE.

As in the case of NMIFE, an NMCFE scheme allows to evaluate \emph{perturbed} $n$-ary functions $f_{\ell}$ on encrypted data, receiving only $f_{\ell}(\xb_{1, \ell}, \cdots, \xb_{n, \ell}) + \noise$.
Stemming from the notion of MCFE, we introduce labels $\ell$ to the ciphertext and each decryption key, ensuring that only ciphertexts and decryption keys with the same label can be combined.
This ensures a fine-grained access structure and is especially helpful to achieve DP, where each function evaluation on a different set of data should contain its own perturbation noise.
In addition, this allows for different vector lengths of $\xb_i$, i.e. different number of attributes, per label, as each function only can be applied to the specific encoded label.

The extension to the dynamic setting is inspired by Chotard et al. \cite{C:CDSGPP20}, where the encryption key generation is separated from the setup phase.
This allows us to dynamically register new clients, eliminating the need for a pre-defined number of clients.
In an DyNMCFE each function $f_{\ell, S}$ is additionally associated with a set $S$ of clients, i.e., it can only be applied to the ciphertexts of these particular clients.
Our definition allows for a polynomial number of clients, which suffices in practice, as the number of data holders is also bounded by, e.g., the people living in a country\footnote{It is important to note that in our concrete scheme presented in  \Cref{sec:con-NMCFE}, this bound can be chosen arbitrarily large (as long as it remains polynomial in the security parameter), without sacrificing any efficiency.}.
A naive way to achieve this with classical NMCFE would be to set the function coefficients for the clients $i \notin S$ to zero, but this would require a  pre-defined number of clients, which could be unpractical.

\subsection{Definition and Correctness}

In the following we give the formal definition of DyNMCFE, where we mark deviations from NMIFE coming from each modification in boxes, i.e., from \graybox{multi-client} and \gbox{dynamic number of clients} respectively.
In other words, if we strip away these changes, it gives us the formal definition of NMIFE as in \cite{ZASH24}.
%

\begin{definition}[(Dynamic) Noisy Multi-Client Functional Encryption]\label{def:nmcfe}
Let $\cF_{\lambda, \nmax}^{\mmax} = \{\cF_{S}^{m_\ell}\}$ be a family of \gbox{sub-families  $\cF_{S}^{m_\ell}$, parameterized by an index} \gbox{ set $S \subseteq [\nmax]$, $\nmax = \poly(\lambda)$,} and \graybox{$\mmax \geq m_\ell \in \N$}, containing functions $f: (\mathcal{X}_{\ell}^{m_\ell})^{|S|} \to \mathcal{Y_\ell}$.
Let \graybox{$\Labels = \{0,1\}^* \cup \{\bot\}$} be a set of labels.
A dynamic noisy multi-client functional encryption scheme (DyNMCFE) of polynomial arity for $\cF_{\lambda, \nmax}^{\mmax}$ and \Labels is a tuple of five efficient algorithms $\mathsf{NMCFE}=(\Setup, \gbox{\EKeyGen,}$ $\Enc, \KeyGen, \Dec)$  of the following form:
\begin{description}
	\item[$\Setup(1^\lambda, \mmax, \nmax)$:]
	Takes as input the security parameter $\lambda$, the maximum vector length $\mmax$ and maximum number of clients $\nmax$.
	It outputs a set of public parameters \pp implicitly defining the function family and a master secret key \msk.
	All remaining algorithms implicitly take \pp.
	\item[\gbox{$\EKeyGen(\msk, i)$}:] Takes as input the master secret key $\msk$, an index $i \in [\nmax]$ and generates a secret encryption key $\ski$. It sets $\msk := \msk \cup {(i, \ski)}$ and returns $\ski$.
	\item[$\Enc(\ski, x_i, $ \graybox{$\ell$}):] Takes as input the encryption key \ski for a slot $i \in [\nmax]$, a message $x_i \in \mathcal{X}_\ell^{m_\ell}$ and \graybox{a label $\ell \in \Labels$}.
	It outputs a ciphertext $\ct_{i,\grayindex{\ell}}$.
	\item[$\KeyGen(\msk,$ \gbox{$S$,}$ f, $\graybox{$\ell$,}$ \distr)$:] Takes as input the master secret key \msk, a function $f$  \gbox{of arity $|S|\subseteq [\nmax]$,} \graybox{a label $\ell \in \Labels$} and a distribution $\distr$ over some subset $\Delta \subseteq \mathcal{Y}_\ell$ such that $\Pr\left[f + \Delta \in \cF_{S}^{m_\ell}\right]=1$.
	Sample $\noise \leftarrow \distr$ and output a decryption key $\dk_{f,\grayindex{\ell}, \gbox{\(\scriptstyle S\)}}$.
	\item[$\Dec(\dk_{f, \grayindex{\ell}, \gbox{\(\scriptstyle S\)}}, \{\ct_{i,\grayindex{\ell}}\}_{i \in \gbox{\(\scriptstyle S\)}})$:] Takes as input the decryption key $\dk_{f, \ell, S}$ and $|S|$ ciphertexts $\{\ct_{i,\ell}\}_{i \in S}$, all encrypted under the same label used for the decryption key.
	It outputs a value $z \in \mathcal{Y_\ell}.$
\end{description}
\end{definition}


The definition is equivalent to its non-dynamic version if all parameters are fixed and encryption keys are generated during setup.
In detail, if $S=[\nmax]$ for all keys, $m_\ell = \mmax$ for all $\ell \in \Labels$ being part of the public parameters, and each $\sk_i$ is generated during setup, this yields NMCFE.


\begin{remark}
	The knowledgeable reader may notice that in contrast to the definition of classical MCFE, the functions are also tied to the labels.
	This is necessary to ensure that only decryption keys under the same label as the ciphertexts are able to decrypt these, which is needed for DP to ensure independent noise in each function evaluation.
	Generally, the label can always be provided to $\KeyGen$, but is ignored inside the algorithm to guarantee that each key can be applied to the ciphertexts independent of their labels.
\end{remark}

\subsubsection*{Correctness}
The correctness definition for an NMIFE scheme says that if all algorithms have been applied correctly, then the distribution of $\Dec(\dk_f, \Enc(x))-f(x)$ should be indistinguishable from $\distr_f$.
This results in obtaining $f(x)+ \noise$ with almost the same probability as the probability that $\noise$ is drawn.
For DyNMCFE, the definition is modified such that decryption keys can only decrypt those ciphertexts which share the same label and are in the corresponding index set.

Formally, we obtain the following definition, which is a straightforward adaptation from~\cite{ZASH24}.

\begin{definition}[Correctness of DyNMCFE\label{def:correctness_unboundedNMCFE}]
	A dynamic NMCFE scheme $\nmcfe=(\Setup, \EKeyGen, \Enc,\KeyGen, \Dec)$ is correct if for any security parameter $\lambda$, maximum number of clients $\nmax = \poly(\lambda)$ for all $S\subseteq [\nmax]$,  $f\in \cF_S^{m_\ell}$, $\mmax \in \N$, $\ell \in \Labels, x_i \in \mathcal{X}_{\ell}^{m_\ell}$ for all $i \in |S|$, $m_\ell \leq \mmax \in \N$ and all distributions $\distr_f$ over some set $\Delta\subseteq\resSpace_\ell$ with $f+\Delta\in \funcFam$,  when  $(\pp, \msk) \leftarrow \Setup(1^\lambda, \mmax)$, $\{\ski \gets \EKeyGen(\msk, i)\}_{i \in S}$ and $\ct_{i,\ell} \leftarrow \Enc(\sk_i, x_i,  \ell)$, it holds:
	\begin{align*}
		\Pr
		\begin{bmatrix}
			\Dec(\KeyGen(\msk, S, f, \ell, \distr_f), \{\ct_{i,\ell}\}_{i \in S}) - f(\{x_i\}_{i \in S})
		\end{bmatrix}
		\equiv \distr_f,
		\end{align*}
	 where the probability is taken over the random coins of the algorithms of \dynmcfe.
\end{definition}
\subsection{Security of NMCFE}
Informally speaking, an NMCFE scheme is considered secure, if the ciphertexts and decryption keys do not reveal any information about the plaintext data and the noise for any label, except for the desired information $f_{\ell}(\{x_{i, \ell}\}_{i \in S}) + \noise$.
Similar as in the definition of DyNMCFE (\cref{def:nmcfe}), \cref{def:securitymcfe} combines the security definitions of NMIFE, MCFE and takes inspiration from \cite{C:CDSGPP20}.

The biggest difference from (N)MIFE to (N)MCFE is the allowance of corruptions, which means that the adversary can corrupt clients and obtain their encryption keys.
Moreover, the admissibility conditions must hold  with respect to each label $\ell$ and each queried subset of clients $S$.

It is necessary that before any decryption key query including clients in $S$ is posed, the clients in $S$ are already registered and have obtained an encryption key.
For this reason, we introduce a dummy oracle $\QEKeyGen$, which generates encryption keys.
We stress that the adversary does not receive an output from that oracle and hence, does not obtain any other encryption keys other than those they corrupted.

For now, we only consider security for one-time schemes, i.e., schemes, which only permit one ciphertext per label and client.
These schemes suffice for our use case as each function evaluation needs its own independent noise for different datasets.
We stress that all concurrent noisy FE schemes also only provide one-time security.

In the following, $\distr_{\noise}$ defines the distribution that outputs \noise with probability one.
Deviations from one-time NMIFE coming from \graybox{MCFE} and the \gbox{dynamic} setting are again marked in the respective boxes.
If a whole algorithm or condition exists due to that change, only its name will be marked.

\begin{definition}[One-time Security of DyNMCFE] \label{def:securitymcfe}
	Consider the DyNMCFE scheme $\dynmcfe$$=(\Setup, \EKeyGen, \Enc,\KeyGen, \Dec)$ for a maximum plaintext length of $\mmax$, maximum number of clients $\nmax$ and label set \Labels.
	For any security parameter $\lambda$,
	consider the following  $\mathsf{IND}_{\beta}^{\dynmcfe}$ game between an adversary \adv and a challenger \challenger.
	The game involves a set \HS of honest clients, initialized to $ \HS := [\nmax] $, a set of corrupted clients \CS, initialized to $\CS := \emptyset$ and a set of queried index sets per label $\cS$, initialized to $\cS := \emptyset$.
	\begin{description}

		\item[Initialization:] At the beginning, \challenger runs $(\pp, ) \leftarrow \Setup(1^\lambda, \mmax, \nmax)$.
		Then, they choose a random bit $\beta \leftarrow \bin$ and hand $\pp$ to \adv.

		\item[\graybox{Corruption queries:}] \adv can pose queries to the corruption oracle  $\QCor(i)$ before any other queries to obtain $\ski \gets \EKeyGen(\msk, i)$.
		\challenger set $\CS := \CS \cup \{i\}$ and $\HS := \HS \setminus \{i\}$.
		Any further query to $\QCor(i)$ is answered with the same $\ski$.

		\item[\gbox{Encryption key queries:}] The adversary \adv can adaptively pose encryption key queries $\QEKeyGen(i)$, for which \challenger generates $\ski \gets \EKeyGen(\msk, i)$ and updates $\msk := \msk \cup \{\ski\}$.
		For any given $i \in [\nmax]$, only one query is allowed and any subsequent query for $i$ is ignored.
		Nothing is returned to \adv.
		\item[Encryption queries:] The adversary \adv can adaptively transmit encryption quer\-ies  $\QEnc(i, x_i^0, x_i^1, $ \graybox{$\ell$}$)$.
		\gbox{If either $i \in \CS$ or $\QEKeyGen(i)$ has been previously} \gbox{called,}
		they are answered by  $\ct_{i,\ell} \gets \Enc(\ski, x_i^\beta,\ell)$ and ignored otherwise.
		For any given pair $(i, \ell)$, only one query with return value is allowed and any subsequent queries involving the same $(i, \ell)$ are ignored.
		\item[Decryption key queries:] \adv can adaptively pose queries $\QKeyGen($\gbox{$S$}$, f, $\graybox{$\ell,$} $\noise^0, \noise^1)$ to receive functional decryption keys.
		If for all $i \in S$, \gbox{$i \in \CS$ or } \gbox{$\QEKeyGen(i)$} has been called,  \challenger returns $\dk_{f, \ell, S} \gets \KeyGen(\msk, S, f, \ell, \distr_{\noise^\beta})$ and \gbox{updates $\cS := \cS \cup {(\ell, S)}$.}
		Otherwise, the request is ignored.
		%
		%
		\item[Finalize:] \adv outputs a bit $\beta' \in \bin$.
		\challenger checks, if \adv acted admissibly.
		If not, \challenger sets $\beta' =0$.
		\adv wins, if $\beta' = \beta$.
	\end{description}

	Depending on the specification of the game defined in advance, we call \adv admissible, if all of the following conditions hold.
	\begin{itemize}
		\item[\graybox{i)}] If $i \in \CS$, then for any query $\QEnc(i, x_i^0, x_i^1, \ell)$, $x_i^0 = x_i^1$, i.e., \adv cannot trivially distinguish  both cases by creating their own ciphertexts.
		\item[ii)] For all $($\graybox{$\ell$,} \gbox{$S$}$) \in \cS$, $\QEnc(i, \cdot, \cdot, \ell)$ has been queried for all $i \in \HS\cap S$, i.e., for each queried function, decryption should be possible.
		\item[iii)] For any label $\ell \in \Labels$, any tuple $(\ell, S) \in \cS$, any family containing all honest encryption queries with respect to $S$, i.e., $\{\QEnc(i, x_i^0, x_i^1, $ \graybox{$\ell$}$)\}_{i \in \gindex{\HS\cap S}}$,
		for any family of inputs $\{x_i \in \mathcal{X}_{\ell}^{m_\ell}\}_{i \in \CS}$,
		any query $\QKeyGen($\gbox{$S$}$, f, $\graybox{$\ell,$} $\noise^0, \noise^1)$,
		we require that:
		\begin{align}\label{eq:admissibility-set}
			f(\{x_i^0\}_{\gindex{i \in S}} ) + \noise^0 = f(\{x_i^1\}_{\gindex{i \in S}}) + \noise^1,
		\end{align}
		i.e., no function evaluation can yield distinction.

	\end{itemize}
	The advantage of $\adv$ in this game is defined as
	\begin{align*}
		\advantage^{\mathsf{IND}}_{\dynmcfe, \adv}(\lambda)
		= &\left| \Pr(\mathsf{IND}_0^{\dynmcfe}(\lambda, \adv) =1) \right.\\
		&-\left.\Pr(\mathsf{IND}_1^{\dynmcfe}(\lambda, \adv) =1) \right|.
	\end{align*}
	A DyNMCFE scheme \dynmcfe provides $\mathsf{IND}$-security, if
	$\advantage^{\mathsf{IND}}_{\dynmcfe, \adv}(\lambda) \leq \mathsf{negl}(\lambda).$
\end{definition}

In particular, conditions i) - iii) are commonly used to prevent trivial attacks through either the possession of encryption keys or through varying function evaluations.
Note that ii) is required such that \cref{eq:admissibility-set} always has to hold, as soon as there is a function query including at least one honest client.
These restrictions are common in FE and are simply modified to fit our use of labels and subsets of clients.

In general, the security definition can be further generalized to allow different modifications. such as adaptive corruptions, i.e., enabling corruptions at any point in time, allowing more than one ciphertext per slot per label or requiring that $\noise^0 = \noise^1$, which would yield regular MCFE.
Such a definition is included in \cref{sec:extended-def} for a non-dynamic NMCFE, as it facilitates our proofs.
Note an DyNMCFE scheme which satisfying Definition \ref{def:securitymcfe} suffices for our use case and any modification, except for the adaptive corruptions, does not give us any benefit.




\section{Privacy-Preserving Protocol}\label{sec:ppfl}
%
In the scenario described in \cref{sec:intro}, we assume that the data is distributed  across multiple data sources.
Via a central trusted entity, analysts and data holders can be connected and the analysis can be performed.
In the following, we propose a protocol for PPML based on DyNMCFE.
To prove the efficiency of the proposed protocol, we train a logistic regression model on medical data  in \cref{sec:implementation} using \lnmcfeot (\cref{sec:noisy-construct}).
Our generic approach is not limited to logistic regression but allows for linear or linearly approximated ML models, which are trained using an iteration-based algorithm, e.g., private GD.
To describe our protocol, we first examine the parties involved.



\paragraph{Parties} We consider three different parties.

\begin{description}
	\item[Authority:]  The authority is a trusted party with limited resources who handles the setup of the DyNMCFE scheme. They provide encryption keys to the data holders and decryption keys to the analyst.
	This could be for example a hospital or an ethical committee that oversees such studies.
	On request of a function $f$, they answer truthfully with a perturbed decryption key $\dk_f$, where the noise \noise is sampled according to a DP providing distribution.
	They keep track of the privacy budget, ensuring that DP is always achieved.
	\item[Data Holder:]  We assume the data is distributed between several parties, e.g., the patients themselves, where each party is in possession an entire record.
	They are willing to submit their individual data to research honestly (possibly at different points in time).
	On the one hand, they still want to maintain their privacy and on the other hand, do not wish to actively participate in the process themselves.
	\item[Analysts:] The analysts want to learn information from the research data, e.g., by training some ML model.
	They can behave maliciously in the sense that they can corrupt data holders, i.e., work with them to obtain their data in plain or even their secret keys, and request any functions supported by the functionality of DyNMCFE.
	As the sampled noise directly depends on the queried function and not on the training of the ML model, the privacy of the data holders remains protected.
\end{description}

\paragraph{Design Considerations}
Although there is extensive research in the area of FE, efficient schemes are still restricted to support only linear or quadratic functions.
For efficiency, we focus on linear schemes.
In fact, a linear scheme suffices for our use case for the following reason.
Remember that in private GD, we want to compute $F_{\thetab}(X)$ to update our gradients, which is the sum over all gradients $\nabla L(\thetab, \xb_i)$, i.e., $F_{\thetab}(X) := \frac{\alpha}{n} \sum_{i \in [n]} \nabla L(\thetab, (\xb_i))$.
Particularly, each $\nabla L$ only depends on the record of one individual.
Hence, if $\nabla L$ can be approximated as a polynomial, we can precompute its monomials to obtain an extended plaintext $\widetilde{\xb_i}$.
As each monomial now has its own slot, $\nabla L$ can be computed as a scalar product between $\widetilde{\xb_i}$ and its respective coefficients, i.e., it can be seen as a linear function.
To ensure that only the necessary monomials are determined and the coefficients are rightfully multiplied with the monomials, we can identify this transformation with a mapping $M: \mathcal{X} \to \mathcal{\widetilde{X}}$.
This way, $F_{\thetab}(X)$ can be correctly evaluated by a linear DyNMCFE scheme.

$M$ influence the size of the plaintext and hence of the ciphertext.
Thus, they have to be known before encryption to verify that the scheme supports the individual plaintexts of that particular size.
In our scheme, this should generally not be a problem as the maximum bit-length is rather large.
Moreover, it is desirable that the overall plaintext size does not grow too much as this directly influences the overall runtime.
This will be further discussed in \cref{sec:implementation}.

\paragraph{Protocol}

Given the three parties and the linearized private GD algorithm, our PPML protocol is as follows.

The authority sets up the environment for the DyNMCFE scheme with a maximum plaintext bit-length $l_{\max}$ and a label space $\Labels$\footnote{Technically, we also require a maximum number of clients $\nmax$. However, as $\nmax= \poly(\lambda)$, we can set it to $2^{34}$, which is more than twice of current world population. Hence, this restriction is ignored in the protocol.}.

Whenever a data holder joins the system, they exchange their privacy budget $(\epsilon_i, \delta_i)$ for the individual encryption key.

Whenever an analyst wants to start a new analysis on data that has not been provided yet,
they send a request to the authority about what data they need and the form of the data, e.g., some normalization factors, realized by $M$.
If the required data is within the limits of the scheme, the authority approves the study and assigns the analysis an unused label $\ell \in \Labels$.
The request is forwarded together with all necessary information to the data holders.

Each data holder can decide, whether they want to participate and tell the authority their decision.
If they do, they encrypt their data according to the requirements and submit it to a server which can be accessed by the analyst.
Although they may submit different data under several labels, possibly over time, they are not involved further in the evaluation process.
It is also possible that a data holder registers at this point in time and participates in the particular study.
This exchange is displayed in \cref{fig:phaseII}.

\begin{figure}
	\begin{center}
	\begin{tikzpicture}
	\footnotesize
	\tikzmath{
		\yStart = -0.75;
		\yEnd = -10;
		\y = \yStart;
		\add = 0.5;
		\user = 0;
		\uR =  \user + 0.125;
		\auth = 4;
		\authL = \auth - 0.125;
		\authR = \auth + 0.125;
		\ana = 8;
		\anaL = \ana - 0.125;
	}
	%

	%
	\node (u) at (\user,0)  {\textbf{Data Holder $i$}};
	\node (auth) at (\auth, 0)  {\textbf{Authority}};
	\node (analyst)  at (\ana, 0)  {\textbf{Analyst}};
	%

	%


	\draw [->] (\anaL, \y) -- (\authR, \y)
	node [above, midway]{Study request $(\mathcal{X}, M)$};
	\tikzmath{ \y = \y -\add; }
	\node (u00) at (\user, \y){};
	\draw [<-] (\anaL, \y) -- (\authR, \y)
	node [above, midway]{Approval with $\ell$};
	\draw [<-] (\uR, \y) -- (\authL, \y)
	node [above, midway]{$(\ell, \mathcal{X}, M)$};
	\tikzmath{ \y = \y -\add; }

	%
	%
	\node (u0) at (\user, \y){Join? If yes: };
	\tikzmath{ \y = \y -\add; }
	\draw [->] (\uR, \y) -- (\authL, \y)
	node [above, midway]{Register};
	\tikzmath{ \y = \y -\add; }
	\node (u1) at (\user, \y){$\widetilde{\xb}= M(\xb)$};
	\node (a01) at (\auth, \y){$I_{\ell} := I_\ell \cup \{i\}$};
	\tikzmath{ \y = \y -\add + 0.125; }
	\node (u15) at (\user, \y){$\ct_i \leftarrow \Enc(\ski, \widetilde{\xb}, \ell)$};
	\tikzmath{ \y = \y -\add +0.25; }
	\node (a11) at (\auth, \y){};

	\tikzmath{ \y = \y -0.25; }
	\draw [->] (\uR, \y) -- (\anaL, \y)
	node [above, midway]{$\ct_i$};
	\node (uend) at (\user, \y){};

	%
	%
	\draw[dashed] (\ana, \yStart) -- (\ana, \y);

	\draw[dashed] (\auth, \yStart) -- (a01);
	\draw[dashed] (a01) -- (a11);
	\draw[dashed] (u00) -- (u0);
	\draw[dashed] (u0) -- (u1);
	\draw[dashed] (u0) -- (u1);
	\draw[dashed] (u15) -- (uend);

\end{tikzpicture}
\end{center}
\caption{Analysis request and data gathering.} \label{fig:phaseII}
\end{figure}
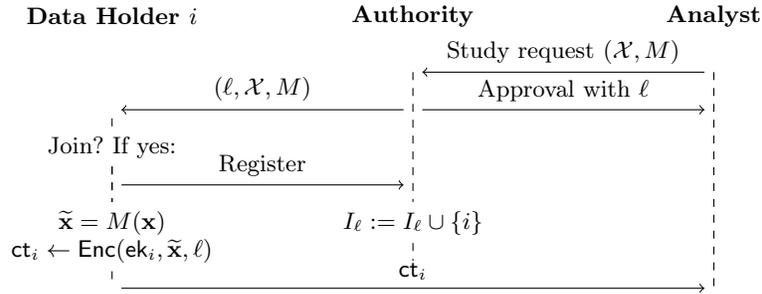

To train an ML model on the gathered ciphertexts using GD, teh analyst first initialize the weights of the model, $\theta^{(0)}$.
In each iteration $t$, they pose queries including the requested functions according to $M$ and the amount of privacy budget in said iteration, i.e., $(F_{\thetab^{(t)}}=\{F_{\thetab^{(t)}}[j]\}_{j \in [0;m]}, (\epsilon^{(t)}, \delta^{(t)}))$.
The authority checks if this would exceed any participant's remaining budget.
If so, they contact the analyst to remove the particular individual's in the next decryption key queries.
Otherwise, they update each participant's privacy budget according to the consumption of privacy and compute $\distr \gets \GM(\epsilon^{(t)}, \delta^{(t)}, \Delta_2(F_{\thetab}^{(t)}))$, which provides DP according to \cref{def:anagm}.
They generate the functional decryption keys with respect to $\distr$ and return them to the analyst.

After each iteration, the analyst updates their model parameters $\thetab$ and starts the next iteration by again requesting functional decryption keys. 
The process ends after a fixed number of iterations $T$ or if the privacy budget of all data holders is exceeded.
\cref{fig:training} shows the training phase if none of the data holders' privacy budgets are depleted before iteration $T$.

Note that all of the decryption keys can be made public for other analysts.
Further, another analyst may conduct a training as in \cref{fig:training} for a different model, depending on the remaining privacy budget.

\begin{figure}
	\begin{center}
		\begin{tikzpicture}
			\footnotesize
			\tikzmath{
				\yStart = -0.5;
				\yEnd = -10;
				\y = \yStart;
				\add = 0.5;
				%
				%
				\auth = 1.25;
				\authL = \auth - 0.125;
				\authR = \auth + 0.125;
				\ana = 8.65;
				\anaL = \ana - 0.125;
			}
			\node (auth) at (\auth, 0)  {\textbf{Authority}};
			\node (analyst)  at (\ana, 0)  {\textbf{Analyst}};

			%
			%
			\node (an1) at (\ana, \y) {Initialize $\thetab^{(0)}$};
			\tikzmath{ \y = \y -\add-0.125; }
			\draw [->] (\anaL, \y) -- (\authR, \y)
			node [above, midway] {$(F_{\thetab^{(0)}}, (\epsilon^{(0)}, \delta^{(0)})) $};
			\tikzmath{ \y = \y -\add; }
			\node (a15) at (\auth, \y) {Check $(\epsilon_i, \delta_i) \forall i \in I_{\ell}$};
			\tikzmath{ \y = \y -\add + 0.125; }
			\node (a2) at (\auth, \y) {$\distr \gets \GM(\epsilon^{(t)}, \delta^{(t)}, \Delta_2(F_{\thetab}^{(t)}))$};
			\tikzmath{ \y = \y -\add + 0.125; }
				\node   at (\auth, \y) {$\forall j \in [0;m]:$};
			\tikzmath{ \y = \y -\add + 0.125; }
			\node (a25)  at (\auth, \y) {$\dk^0_j \leftarrow \KeyGen(\msk, I_\ell, F_{\thetab^{0}}[j], \ell, \distr[j])$ };
			\tikzmath{ \y = \y -\add+0.125; }
			\node (a3)  at (\auth, \y) {$(\epsilon_i, \delta_i) := (\epsilon_i - \epsilon^{(0)}, \delta_i - \delta^{(0)}) \forall i \in I_{\ell}$};
			\tikzmath{ \y = \y -\add-0.2; }
			\draw [<-] (\anaL, \y) -- (\authR, \y)
			node [above, midway]{$\{\dk_j^0\}_{j \in [0;m]}$};
			\tikzmath{ \y = \y -\add; }
			\node (an2) at (\ana, \y) {$\forall j \in [0;m]:$};
			\tikzmath{ \y = \y -\add+ 0.125; }
			\node (an3) at (\ana, \y) {$\thetab^{(1)}[j] \leftarrow \Dec(\dk_j^{0},\{\ct_i\}_{i \in I_\ell})$};
			\tikzmath{ \y = \y -\add-0.125; }
			\draw [->] (\anaL, \y) -- (\authR, \y)
			node [above, midway] { $(F_{\thetab^{(1)}}, (\epsilon^{(1)}, \delta^{(1)})) $};
			\tikzmath{ \y = \y -\add; }
			\node (a4) at (\auth, \y){};
			\node (an4) at (\ana, \y){};
			\tikzmath{ \y = \y -\add+0.25; }
			\node at (\auth, \y) {$\cdots$};
			\node  at (\ana, \y) {$\cdots$};
			\tikzmath{ \y = \y -\add+0.25; }
			\node (a5) at (\auth, \y){};
			\node (an5) at (\ana, \y){};


			%
			\draw [<-] (\anaL, \y) -- (\authR, \y)
			node [above, midway]{$\{\dk_j^{\ep -1}\}_{j \in [0;m]}$};
			\tikzmath{ \y = \y -\add; }
			\node (an6) at (\ana, \y) {$\forall j \in [0;m]:$};
			\tikzmath{ \y = \y -\add+ 0.125; }
			\node at (\ana, \y) {$\thetab^{(\ep)}[j] \leftarrow \Dec(\dk_j^{\ep-1},\{\ct_i\})$};
			%
			%
			\draw[dashed] (an1) -- (an2);
			\draw[dashed] (an3) -- (an4);
			\draw[dashed] (an5) -- (an6);
			\draw[dashed] (\auth, \yStart) -- (a15);
			\draw[dashed] (a3) -- (a4);
			\draw[dashed] (a5) -- (\auth, \y);
			%
		\end{tikzpicture}
	\end{center}

	\caption{Training phase.} \label{fig:training}
\end{figure}
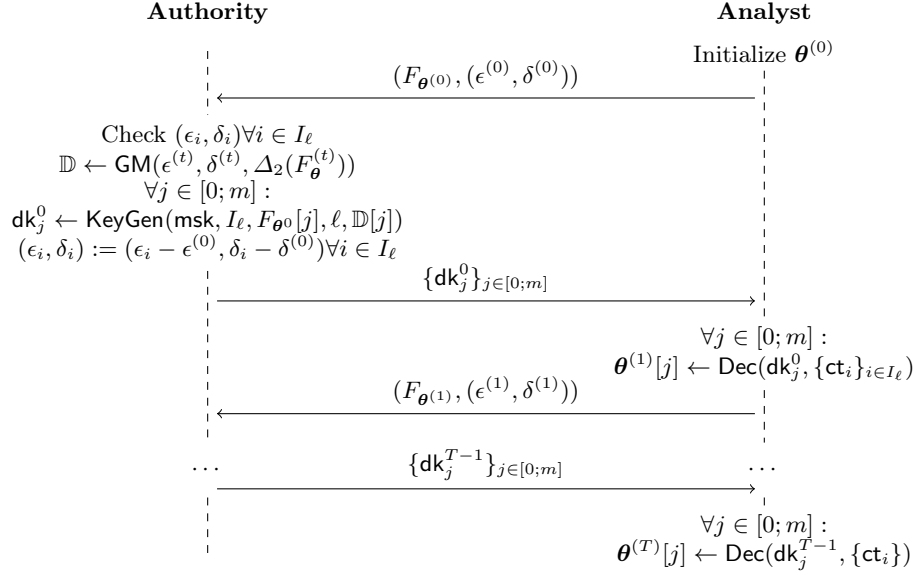

\paragraph{Threat Model and Privacy Analysis}
We require that the authority honestly generates the decryption keys and takes care of the privacy budget.
As the authority tracks the privacy budget of the data holders over all labels and function queries, DP and hence output privacy is always guaranteed.
Note that at least one trusted party is always required, especially in protocols using HE or MPC. LDP would be the only approach without a trusted party, but distorts the utility of the trained model, as we show in \cref{sec:implementation}.

The analyst on the other hand may collude with data holders.
More precisely, they may know their data and even their particular encryption keys.
The privacy of the non-colluding data holders is still guaranteed due to DP and the allowance of corruptions of the DyNMCFE scheme.

Moreover, the analyst is not limited to functions which help them train a ML model.
However, due to the assurance of DP, the analyst would only jeopardize the usefulness of their model by asking for other decryption keys.

Overall, by employing  this approach, both input and output privacy of all individuals are ensured. 

\section{\lnmcfeot: A concrete DyNMCFE instantiation}\label{sec:noisy-construct}
In this section, we present our one-time DyNMCFE scheme, \lnmcfeot .
The main idea of \lnmcfeot is similar to the one-time MIFE scheme of Abdalla et al.~\cite{PKC:ABKW19}, which makes use of the linearity of the one-time pad.
We extend their scheme to support noise-hiding, i.e., to be a secure NMIFE scheme.
Let us first recall their scheme and our extension to noise-hiding, before we present \lnmcfeot.

\subsection{Warm up: \nmifeot}\label{sec:mifeot}
In the following we refer to the scheme of  Abdalla et al.~\cite{PKC:ABKW19} as \mifeot. Extensions implemented to lift \mifeot to a noisy variant \nmifeot are marked in boxes.
If $\distr$ is chosen to be the all-zero distribution, both schemes are equivalent.

\begin{construction}[\nmifeot]\label{def:nmifeot}
	Let \(\cF_{q,n}^m\) be the class of multi-input inner products over \(\Z_q\). \nmifeot scheme for \(\cF_{q,n}^m\) consists of the following algorithms:
	\begin{description}
		\item[\(\Setup^{\textnormal{ot}}(1^\lambda, m, n)\):]
		On input of the security parameter $\lambda$, the  vector length $m$ and number of clients $n$, set $\msk := \emptyset$.
		For all $i \in [n]$, run $\ski \gets \EKeyGen^{\textnormal{ot}}(\msk,i)$ and return $\msk$.

		\item[$\EKeyGen^{\textnormal{ot}}(\msk, i)$:] On input of the master secret key $\msk$,
		sample \(\ski \sampler \Z_q^m\).
		Set  $\msk := \msk \cup \{(i, \ski)\}$ and return \(\sk_i\).

		\item[\(\Enc^{\textnormal{ot}}(\sk_i, \xb_i)\):] On input of the encryption key \(\ski\) and message \(\xb_i \in \Z_q^m\) for slot \(i \in [n]\), return \(\ct_i := (i, \mat c_i:= \xb_i + \ski \mod q)\).

		\item[\(\KeyGen^{\textnormal{ot}}(\msk, f_{\yb}, \fbox{$\distr$})\):]
		On input of the master secret key  $\msk = \{(i,\ski)\}_{i\in [n]}$, a function $f_{\yb}$ defined through \(\yb = (\yb_1, \dots, \yb_n)\) and a \fbox{distribution $\distr$ over} \fbox{ \(\Z_q\), sample \(\noise \leftarrow \distr\).} Return \(\dk_{\yb} =(\yb, z := \sum_{i \in [n]} \langle \ski, \yb_i \rangle \fbox{- \noise} \mod q)\).

		\item[\(\Dec^{\textnormal{ot}}(\dk_{\yb}, \ct_1, \dots \ct_n)\):]
		On input of the decryption key \(\dk_{\yb} = ((\yb_1, \dots, \yb_n), z)\) and ciphertexts $\ct_i = (i, \mat c_i)$, $i \in [n]$, return \(\sum_{i \in [n]} \langle \mat c_i, \yb_i \rangle -z \mod q\).
	\end{description}
\end{construction}

\paragraph{Correctness} The correctness of the scheme in $\Z_q$ follows directly by the way $z$ is chosen and that $\sum_{i \in [n]} \langle \mat c_i, \yb_i \rangle = \sum_{i \in [n]} \langle \xb_i + \sk_i, \yb_i \rangle$.
Thus, the result after decryption is $f(\xb_1, \dots, \xb_n) + \noise$, whose distribution is solely determined by $\distr_f$.
\begin{remark}
For correctness in $\Z$, the modulus $q$ has to be chosen large enough.
In particular, assume $\max_{j \in [m]} \xb_i[j] < X$ for all $i \in [n]$, $\max_{j \in [mn]}\yb[j] < Y$ and $\Pr(\noise > d) \leq \negl(\lambda)$, where $\noise \leftarrow \distr_{f}$, for all queried $f$.
Then \nmifeot is correct in $\Z$, if
\begin{align}\label{eq:correctness}
	nmXY + d < q.
\end{align}
\end{remark}

\paragraph{Security}
For our security proof, we use the information theoretical security of \mifeot as described above.
In a sequence of games we transform \nmifeot to \mifeot, which concludes the following theorem.
Its proof can be found in Appendix~\ref{sec:security-nmife}.
\begin{theorem}\label{thm:sec-nmifeot}
	The \nmifeot scheme presented in Construction~\ref{def:nmifeot} is $\mathsf{IND}$-secure with $\advantage_{\nmifeot, \adv}^{\mathsf{IND}}(\lambda) = 0$ for any $\ppt$ adversary \adv.
\end{theorem}

\subsection{\lnmcfeot: An efficient DyNMCFE scheme}\label{sec:con-NMCFE}

Recall that the difference between NMIFE and DyNMCFE is the use of labels and the ad-hoc registration of clients.
Thus, we need the encryption keys $\sk_i $ to be dependent on the label.
As we additionally require that the keys are labeled as well, we can simply replace the one-time pad with a PRF.
This PRF takes as input the encryption key $\ski$ of client $i$ and the label $\ell$.
If $\ski$ is sampled uniformly at random, this function is computationally indistinguishable from a truly random function. A formal definition is given in \cref{def:PRF}.

For any PRF \PRF of output length $m$, denote by $\PRF_{m'}$ the PRF which only outputs the first $m' \leq m$ values.
This transformation is straightforward by cutting for example the last $m - m'$ values.

\begin{construction}[\lnmcfeot] \label{def:lnmcfeot}
	Let \(\{\cF_{q,S}^{\mmax}\}_{S \subseteq [\nmax]}\) be the family of sub-families of multi-input inner-products over \(\Z_q\) with maximum attributes $\mmax$ and maximum number of clients $\nmax$.
	Let $\Labels$ be the supported label space and $\PRF: \bin^\lambda \times \Labels \to \Z_q^{\mmax}$.
	The \lnmcfeot scheme for \(\{\cF_{q,S}^{\mmax}\}_{S \subseteq [\nmax]}\) consists of the following algorithms:
	\begin{description}
		\item[\(\Setup(1^\lambda, \mmax, \nmax)\):]
		On input of the security parameter $\lambda$, maximum vector length $\mmax$ and maximum number of clients $\nmax$,
		set \(\msk := \emptyset\). Output $\msk$ and the public parameters $\pp = (\mmax, \nmax, \Labels, \PRF)$.

		\item[$\EKeyGen(\msk, i)$:] On input of the master secret key and slot $i \in [\nmax]$, sample \(\ski \sampler \{0,1\}^\lambda\).
		Update \(\msk := \msk \cup \{(i,\sk_i)\}\) and return $\ski$.
		\item[\(\Enc(\sk_i, \xb_i, \ell)\):]
		On input of the secret key $\ski$, a plaintext $\xb_i \in \Z_q^m$ with $m \leq \mmax$,  for slot $i \in [\nmax]$ and a label $\ell \in \Labels$,
		calculate $\zeta = \PRF_m(\ski, \ell)$ and return \(\ct_{i,\ell} = (i, \ell, \mat c_i := \xb_i + \zeta \mod q)\).

		\item[\(\KeyGen(\msk, S, f_{\yb}, \ell, \distr)\):]
		On input of the master secret key \msk, a set of indices $S \subseteq [\nmax]$, a function $f_{\yb}$, a label $\ell \in \Labels$ and a noise distribution $\distr$, if there exists $(i, \cdot) \notin \msk$ for some $i \in S$, return $\bot$.
		Otherwise,
		sample \(\noise \leftarrow \distr\) and calculate $\zeta_i = \PRF(\ski, \ell)$ for all \(i \in S\).
		Return \(\dk_{\yb,\ell, S} =(\ell, \yb, z := \sum_{i \in S} \langle \zeta_i, \yb_i \rangle - \noise \mod q)\).
		\item[\(\Dec(\dk_{\yb, \ell, S}, \{\ct_{i,\ell}\}_{i \in S})\):]
		On input of the decryption key \(\dk_{\yb, \ell, S}=(\ell, \{\yb_{i}\}_{i \in S}, z)\) and ciphertexts $\{\ct_{i,\ell}= (i, \ell, \mat c_i)\}_{i \in S}$,  return \(\sum_{i \in [S]} \langle \mat c_i, \yb_i \rangle -z \mod q\).
	\end{description}
\end{construction}

Note that the construction can serve as an alternative to \nmifeot for a single label if we aim to minimize the size of the encryption keys and set $\nmax =n, \mmax=m$.
This is especially of interest if the data to be encrypted is large.
The size of a ciphertext is $m \log q$ bits, both in \nmifeot and \lnmcfeot, which matches the size of a encryption key in \nmifeot, whereas the encryption key in \lnmcfeot is of size $\lambda$.
Thus, if $m \log q > \lambda$, there is a space savings, though this comes with increased computational costs from the PRF.
In practical applications, we find that $q$ tends to grow rapidly, which enhances the effectiveness of \lnmcfeot.
There are two primary reasons for this.
Firstly, we require correctness in $\Z$ rather than $\Z_q$, necessitating that $q$ be sufficiently large, specifically $q > nmXY + d$, where $X$, $Y$, and $d$ provide bounds on the ciphertexts, functions, and noise, respectively.
Secondly, to convert real-world data in $\R$ to data in $\Z$, a common technique is to employ fixed-point arithmetic using a scaling, as elaborated in Section~\ref{sec:ppfl}.
As $X$ and $Y$ may need to be quite large to achieve the desired precision this directly effects $q$ and therefore the key sizes for \nmifeot.

\begin{remark}\label{rmk:2}
	In a similar manner as Abdalla et al. \cite{C:ACFGU18}, \lnmcfeot (respectively, \nmifeot) can be utilized to transform any single input IPFE scheme, fulfilling certain properties, to a secure label-keyed NMCFE (respectively, NMIFE) scheme, which allows to encrypt multiple messages per client per label.
	The only difference we have to make is to exchange \mifeot with \lnmcfeot (respectively, \nmifeot).

	Most inner-product scheme based on lattices and bilinear groups fulfill the required properties, e.g.,~\cite{C:AgrLibSte16, PKC:MKMS22}.
	This gives us an even broader family of noise-hiding functions which do not rely on function-hiding and can be build from lattices.
	As we are only interested in one-time schemes for DP, we will not go further into this topic, but leave it as an interesting observation.
	\end{remark}

\paragraph{Correctness}
Similar as in the case of \nmifeot, one can directly see that the scheme is correct in $\Z_q$, if the \PRF in use is deterministic.
Correctness in $\Z$ follows if \cref{eq:correctness} is fulfilled.

\paragraph{Security}

The security of \lnmcfeot relies on the security of the underlying PRF as well as \nmifeot.
We focus on deterministic PRFs, as these are more efficient and more common in practice.

\begin{theorem}\label{thm:sec-nmcfeot}
	If \nmifeot from Construction~\ref{def:nmifeot} is $\mathsf{IND}$-secure and \PRF is secure, then \lnmcfeot from Construction~\ref{def:lnmcfeot} is $\mathsf{IND}$-secure. In particular, for any $\PPT$ adversary \adv,
	 there exist $\PPT$ adversaries $\cB_1$ and $\cB_2$ such that
	\begin{align*}
		\advantage^{\mathsf{IND}}_{\lnmcfeot, \adv}(\lambda) \leq 2 \nmax \ql \cdot \advantage_{\PRF, \cB_1}(\lambda)+ \ql \cdot \advantage^{\mathsf{IND}}_{\nmifeot, \cB_2}(\lambda),
	\end{align*}
	where \ql denotes the number of distinct labels queried to \QEnc and \QKeyGen.
	\end{theorem}

The theorem can be proven by straightforward reduction to the non-dynamic variant, where all secret keys are sampled during the setup phase.
On the other hand, for the non-dynamic variant, we show that it suffices to show security for one label.
From there on, we only have to rely on the security of the PRF, i.e., \PRF can only be distinguish from a truly random function with negligible advantage. At this point, an adversary against \nmifeot can perfectly simulate the view of an attacker against the modified scheme.
Since \nmifeot is information theoretically secure,  the \IND-security of \lnmcfeot follows.
 The whole proof can be found in \Cref{sec:sec-lnmcfeot}.

\section{Implementation}\label{sec:implementation}
We provide an implementation of our DyNMCFE scheme  \lnmcfeot employed in Section~\ref{sec:con-NMCFE}.
All experiments were performed on a system running Ubuntu 22.04.2 LTS, 256GB RAM and 18 vCPUs (AMD Epyc 7272).

As real world data usually takes values in $\R$ instead of  $\Z_q$, we make use of fix-point arithmetic.
More precisely, a value $x \in \R$ is converted to $\lfloor x \cdot s\rceil$ for an appropriate scaling factor $s$ and an accordingly chosen $q$.
This scaling factor determines for example the level of precision.
Choosing it large enough is especially important in the presence of DP as the noise also has to be in $\Z_q$.
Normally, this value is rather small.
However, depending on the rounding of the noise, the added noise may be larger than necessary for DP, which reduces the utility of the model.
Note that scaling or any normalization process can be considered as part of $M$ and do not influence the overall protocol presented in \cref{sec:ppfl}.

In order to achieve DP, we use the analytic Gaussian mechanism, as displayed in Definition~\ref{def:anagm}, which is included in the Google Go library~\cite{googleDP}.
Their implementation of the analytic gaussian mechanism employs binomial random variables and appropriate rounding to argue for DP on integers.

As our PRF, on which the security of \lnmcfeot is based on, we choose AES-256 as implemented in the Go crypto library~\cite{go-crypto-aes}.
This is a standard protocol recommended by NIST~\cite{nist} and plausibly post-quantum secure~\cite{rao2017aes}.

Moreover, we consider the non-dynamic variant for comparison, i.e., the encryption keys are generated as part of the setup and all clients part of the analysis.
The overall runtime and communication costs are almost the same for the dynamic version, except that there may be data holders which never participate in any analysis, but query an encryption key.
In addition, we did not include the request for the analysis and the participation check, as both are implicitly present in any protocol, including the non-dynamic NMCFE scheme.
However, these should be negligible.

\subsection{Benchmarks}
We compare the runtime of  \lnmcfeot with the schemes of Zalonis et al.~\cite{ZASH24}, $\mathsf{DiffPIPE}$, and Escobar et al.~\cite{escobar2024computational}, $\mathsf{RIPFE}$.
These are currently the only (noisy) MIFE schemes providing inner-product functionality and supporting DP.

Let $N= nm$ be the size of the dataset.
Table ~\ref{tab:runtime} shows the runtime of all three schemes on the same datasets for different dataset sizes $N$, where $\mat x_i \in [0,2^{16}]^{m}$ for all $i \in [n]$ and $\yb \in [0, 2^7]^{N}$ were sampled randomly.
$\mathsf{RIPFE}$  only supports one client, i.e., $n =1$, so we set $m = N$.
Since $\mathsf{DiffPIPE}$ is pairing based, it becomes very slow for big $m$.
$\mathsf{DiffPIPE}$ is originally a multi-input scheme, therefore to better compare runtimes, we set $n=m=\sqrt{N}$.
We evaluate our scheme for both partitions of the dataset and achieve better runtimes in all algorithms.
Even for larger datasets of size $10^6$, we are still in the range of seconds and milliseconds, respectively.
While running $\mathsf{DiffPIPE}$ with $N = 10^6$, we needed to cancel the setup phase on our server after several days as the RAM was exhausted.

\begin{table*}
    \caption{Runtime comparison of all FE schemes supporting DP. Since $\mathsf{DiffPIPE}$ did not terminate after several days for $N = 10^6$, we are not able to present runtimes. }\label{tab:runtime}
    \centering
		\begin{tabular}{|c|c|c|c|r l|r l|r l|r l|}
        \hline
        $\boldsymbol N$ & $\boldsymbol n$ & $ \boldsymbol m$ &  \textbf{Scheme} & \multicolumn{2}{|c|}{$\boldsymbol\Setup$} & \multicolumn{2}{|c|}{$\boldsymbol\Enc$} & \multicolumn{2}{|c|}{$\boldsymbol\KeyGen$} & \multicolumn{2}{|c|}{$\boldsymbol\Dec$}\\
        \hline
        \multirow{4}{*}{$10^2$}& $1$ & $ 10^2$  &   $\mathsf{RIPFE}$ \cite{escobar2024computational}&3.50 & s & 0.25 &s & 1.89 &ms & 0.14 & s \\
            &$1$ &$ 10^2$ &  \multirow{2}{*}{  \lnmcfeot}& 0.01 & ms&0.12 &ms&0.14 &ms&0.01 &ms \\
            & $10$ & $10$ & & 0.01& ms&0.09 &ms&0.11 &ms&0.02 &ms \\
            & $10$ & $10$ &   $\mathsf{DiffPIPE}$ \cite{ZASH24}&37.25& ms&5.55& ms&77.12& ms&1.06& s  \\
        \hline

        \multirow{4}{*}{$10^4$}& $1$  & $ 10^4$ &   $\mathsf{RIPFE}$ \cite{escobar2024computational} & 3.51 &s &24.48 &s &17.83& ms & 1.51& s\\
        &$1$ & $ 10^4$ &  \multirow{2}{*}{ \lnmcfeot} & 0.01 &ms&4.23& ms&3.13 &ms&1.31 &ms\\
        & $10^2$ & $10^2$ & &0.06& ms&2.27& ms&0.77& ms&1.28 &ms \\
        & $10^2$ & $10^2$ &    $\mathsf{DiffPIPE}$ \cite{ZASH24} &57.32& s&0.23& s&5.64& s&25.20& s \\
        \hline

       \multirow{4}{*}{$10^6$} &$1$ & $ 10^6$ &    $\mathsf{RIPFE}$ \cite{escobar2024computational}&3.59&s & 41.03 & min& 2.23&s & 33.09 &s\\
       & $1$ &$ 10^6$ & \multirow{2}{*}{  \lnmcfeot} & 0.019& ms&0.31& s&0.24 &s&75.41& ms \\
       & $10^3$ & $10^3$ & & 0.47 &ms&50.93& ms&30.48& ms&85.45 &ms \\
       & $10^3$ & $10^3$ &    $\mathsf{DiffPIPE}$  \cite{ZASH24}&  -& & -&&-&&-&   \\
       \hline
    \end{tabular}

\end{table*}

\subsection{Logistic Regression}
Besides benchmarking, we also train a logistic regression using \lnmcfeot.
We only consider the data exchange (including providing the encryption keys) and the training process, as any other costs can be seen as negligible and appear in any other protocol as well.
For the sake of simplicity, we assume that all data holders share the same privacy budget.

As a baseline, we compare the utility to LDP, as this is the only primitive that neither requires the active participation of the data holders nor one or multiple trusted parties with a lot of computational power.

\paragraph{Model}
Let $X = \{\xb_i\}_{i \in [n]}$ be a dataset where each $\xb_i \in \R^{m+2}$ includes $m$ real features, a binary label $\xb_i[m+1] \in \bin$, and a constant $\xb_i[0] = 1$.
For the sake of readability, we define in the following $z_i^{} := \sum_{j \in [0;m]} \thetab^{}[j]\xb_i[j]$.
The goal of a logistic regression training is to find optimal  $\thetab^* \in \R^{m+1}$, e.g., using private GD, such that given a fresh $\xb$, the classification $y$ can be predicted with high probability. More precisely, it should hold that
	$\lfloor \sigma(z_i)\rceil =\xb_i[m+1]$
with high probability, where $\sigma(x) = \frac{1}{1+\exp{(-x)}}$ is the sigmoid function.

Let \ep be a predefined number of iterations, $\alpha$ the learning rate and $\thetab^0 \in \R^{m+1}$ randomly chosen.
In a logistic regression, each update of the parameters, i.e., (\ref{eq:privgrad}), is given as
\begin{align}
    \thetab^{(t+1)}[j]:= \thetab^{(t)}[j]+\frac{\alpha}{n}\sum_{i \in [n]}\left(\xb_i[m+1] - \sigma(z_i^{(t)})\right)\xb_i[j]\label{eq:theta-j}.
\end{align}

\paragraph{Inner-product Transformation}
Since \lnmcfeot only supports inner-product functionality we need to adapt the training algorithm.
A naive approach would be to view~\cref{eq:theta-j} as two nested inner-products.
First, $z_i^{(t)}$ is calculated and plugged into $\sigma$.
Then, $\sigma(z_i^{(t)})$ is regarded as the function vector for $\xb_i$ to calculate the outer sum.

However, this approach requires to add noise twice in each iteration to preserve privacy, which is too much noise to ensure the convergence towards optimal model parameter.
To reduce the amount of noise,~(\ref{eq:theta-j}) can be calculated directly by the underlying DyNMCFE scheme, if we linearize it.

Conveniently, the sigmoid function can be approximated by the least square polynomial $g(x)$ of degree 3 over the domain $[-8,8]$~\cite{kim2018secure}.
More precisely, for
    $$g(x)= -a_1 x^3 +a_2x +0.5$$
with $a_1 = 0.81562/8^3$ and $a_2=1.20096/8$, it holds that $g(x) \approx \sigma(x)$ for $x \in [-8,8]$.

Replacing $\sigma(x)$ with $g(x)$ in~\cref{eq:theta-j} yields a polynomial $f(x)$ of degree 4.
The monomials of $f$ consist only of variables from the same $\xb_i$, which means that all of them can be precomputed by client $i$.
Thus, as established in \cref{sec:ppfl}, instead of encrypting just $\xb_i$, client $i$ precomputes all unique monomials up to degree 4 and encrypts them as well, yielding $\widetilde{\xb}_i$.
In total, as we only have to consider unique monomials, this extends the number of values in each ciphertext to $\widetilde{m} = \frac{1}{24}m^4 + \frac{5}{12}m^3 + \frac{35}{24}m^2 + \frac{37}{12}m +2$.
This allows us to capture the update function as a inner-product with the expanded encrypted values.
The coefficients for the function can be found in Appendix~\ref{ap:logreg}.


\begin{remark}
As $\widetilde{m}$ grows polynomial in $m$, this method is only suitable for datasets with relatively small number of attributes.
For a large number of $m$, we may have to resort to a linear approximation of the sigmoid function ($g'(x) = 0.5 +  0.25x$).
Then, the number of values to be encrypted shrinks down to $\widetilde{m} = 0.5\cdot(3m+m^2)$.
Although this number is still quadratic, the only existing MIFE scheme which supports quadratic functions also requires to encrypt $(\xb \otimes \xb, \xb, 1)$ to obtain truly quadratic functions~\cite{agrawal2022multi}, and they lack to protect the intermediate results by incorporating noise.
\end{remark}

To ensure that $(\epsilon, \delta)$-DP is provided after the last iteration, we set this as our privacy budget  $(\epsilon_{\max}, \delta_{\max})$.
In each iteration, a fraction of this privacy budget is consumed without exceeding $(\epsilon_{\max}, \delta_{\max})$ in total.
The privacy budget could be divided equally among all iterations, i.e., each iteration fulfills $(\epsilon_{\max}/\ep, \delta_{\max}/\ep)$-DP.
However, to achieve a better convergence, we dynamically adapt the the privacy budget consumed in each iteration by spending less at the beginning and more towards the end, as it is also done in~\cite{du2021dynamic}.
In other words, in the first iterations, we statistically have larger noise values, which can still be smoothed out in the later iterations, where statistically smaller noise values are added.

To determine the amount of noise, we must look at the sensitivity of the update function.
Although each of the coefficients of $\thetab^{(t)}$ is calculated separately in the GD, the whole iteration can be seen as evaluating a function with an $(m+1)$-dimensional output.
Without loss of generality, we may assume $\xb_i \in [0,1]^{m+2}$.
Let $\Theta = \sum_{j \in [0;m]} |\thetab[j]|$.
Then, the $l_2$-sensitivity of each iteration is bounded by
\begin{align}\label{eq:sensitivity-logreg}
    \Delta(F_{\thetab}) \leq  \sqrt{m+1}\frac{\alpha}{n}(1+ |a_1 \Theta^3 - a_2  \Theta|).
\end{align}

\paragraph{Results}
We train a logistic regression on four standard benchmarking datasets from the medical context: the Low Birth Weight Study (LBW)~\cite{LBWstudy}, Prostata Cancer Study (PCW)~\cite{pcsStudy}, Umaru Impact Study (UIS)~\cite{uisStudy} and Nhanes III~\cite{nhanes}. Table~\ref{tab:logreg} shows the dataset sizes and the total runtime for 50 iterations.
We did not consider any optimizations of the GD such as for example batching, which could further improve the overall runtime.

The communication costs for each dataset are displayed in \Cref{tab:size_calc}.  The scaling factor is set to be  $s = 10^6$.
For LBW, PCS and UIS, we set the modulus $q=2^{64}$ and for NHANES III, $q= 2^{72}$.
We noticed that $y \in [-5,5]$, i.e., each queried function slot consists of $\lceil \log_2(10 \cdot s) \rceil$ bits.
The rest of the table can be computed according to \cref{tab:size}.

\begin{table}
    \caption{Runtime for training of the logistic regression with 50 iterations on different datasets in minutes.}\label{tab:logreg}
    \centering
    \begin{tabular}{|c|c|c|c|}
        \hline
        \textbf{Dataset} & $\boldsymbol{n}$ & $\boldsymbol{m}$ &  \textbf{Total runtime}\\
        \hline
        LBW & 189 & 10  & 0.19 min \\
        PCS &  380 & 8  & 0.24 min \\
        UIS & 575 &8  & 0.33 min\\
        NHANES  & 16 427 & 11 &42.29 min\\
        \hline
    \end{tabular}
\end{table}

\begin{table}
    \caption{Package sizes for the encryption process per client, per iteration in the training process and in total after 50 iterations, including the encryption key and ciphertext exchange.}\label{tab:size_calc}
    \centering
    \begin{tabular}{|c|l  r|l  r|c|}
        \hline
        \multirow{2}{*}{\textbf{Dataset}} & \multicolumn{2}{|c|}{$\textbf{Encryption}$} & \multicolumn{2}{|c|}{$\textbf{GD Iteration}$} & \multirow{2}{*}{\textbf{Total}}\\
        &\multicolumn{1}{|c}{$\ski$} &\multicolumn{1}{c|}{$\ct$} &\multicolumn{1}{|c}{$f$} & \multicolumn{1}{c|}{$\dk_f$}& \\ \hline
		LBW & 256 b & 7.9  KB& 32.61 KB& 88 B & \multicolumn{1}{r|}{3.06 MB}\\
		PCS & 256 b & 3.93  KB & 13.29  KB & 72  B& \multicolumn{1}{r|}{2.13  MB}\\
		UIS & 256 b& 3.93 KB & 13.29  KB & 72  B & \multicolumn{1}{r|}{2.88 MB}\\
		NHANES  & 256 b&12.10  KB & 48.41  KB & 108  B & \multicolumn{1}{r|}{197.02 MB}\\

       \hline
    \end{tabular}
\end{table}

\Cref{fig:benchmarkDS} displays the utility compared to LDP in form of accuracy of the first three datasets for a fixed number of iterations $\ep =50$ in dependency of $\epsilon_{\max}$.
The parameter $\delta_{\max}=1/n$ is fixed.

The utility highly depends on the dataset, the privacy budget and the learning rate, which can be also seen through the LDP values.
For small epsilon, the noise may be too large such that convergence is not possible, which we marked by setting the accuracy to 0.
We stress that by simply guessing the dependent variable, an accuracy of 0.5 is always possible, which can be also achieved by a random model.
Therefore, LDP is not better than guessing for small epsilon.


 \begin{figure}
	\centering
	\pgfplotstableread{data/accuracy50.dat}{\table}
	\pgfplotstableread{data/accuracyNhanes.dat}{\table}

	\begin{minipage}[t]{0.45\textwidth}
	\begin{tikzpicture}[scale=0.7]
		\begin{axis}[
			xmin = 0.01,
			xmax = 8,
			xtick distance = 1,
			xlabel = {$\epsilon_{\max}$},
			ymin = 0.4, ymax = 0.77, ytick distance = 0.05,
			ylabel={accuracy},
			ytick={0.4, 0.45, 0.5, 0.55, 0.6, 0.65, 0.7, 0.75}, 
			yticklabels={0.4, , 0.5, , 0.6, , 0.7, }, 
			legend pos = south east,
			legend style={  font=\footnotesize  },]

			\addplot[cb1] table [x = {eps}, y = {LBW}] {\table};
			\addplot[cb1,dash pattern=on 0.5pt off 1pt,  line cap = round] table [x = {eps}, y = {LBW_LDP}] {\table};
			\addplot[domain = 0:10, cb1, dashed] { (0.7195767195767195)};
			\addplot[cb2] table [x = {eps}, y = {PCS}] {\table};
			\addplot[cb2,dash pattern=on 0.5pt off 1pt, line cap = round] table [x = {eps}, y = {PCS_LDP}] {\table};
			\addplot[domain = 0:10, cb2, dashed ] { ( 0.7368421052631579)};
			\addplot[cb3] table [x = {eps}, y = {UIS}] {\table};
			\addplot[cb3,dash pattern=on 0.5pt off 1pt,  line cap = round] table [x = {eps}, y = {UIS_LDP}] {\table};
			\addplot[domain = 0:10, cb3, dashed ] { (0.7443478260869565)};
			\legend{ LBW,,,  PCS,,, UIS,,}
		\end{axis}
	\end{tikzpicture}
	\caption{Model utility for 50 rounds, in dependency of $\epsilon_{\max}$. The dashed line shows peak accuracy on plaintext after 500 rounds. Dotted lines represent the maximal accuracy achieved with LDP with 500 rounds.}\label{fig:benchmarkDS}
	\end{minipage}
	\hfill
	\begin{minipage}[t]{0.45\textwidth}
		\begin{tikzpicture}[scale=0.7]
			\begin{axis}[
				xmin = 0.01,
				xmax = 8,
				xtick distance = 1,
				xlabel = {$\epsilon_{\max}$},
				ymin = 0,
				ymin = 0, ymax = 0.89, ytick distance = 0.1,
				ytick = {0, 0.1, 0.2, 0.3, 0.4, 0.5, 0.6, 0.7, 0.8 }, 
				yticklabels={0.0, , 0.2, , 0.4, , 0.6, , 0.8, }, 
				legend pos = south east,
				legend style={  font=\footnotesize  },
				]

				\addplot[cb1] table [x = {eps}, y = {N50}] {\table};
				\addplot[  mark=diamond*, mark options = cb1] coordinates {(8, 0.7470627625251111)};
				\addplot[cb2] table [x = {eps}, y = {N100}] {\table};
				\addplot[mark=diamond*, mark options = cb2 ]coordinates { (8, 0.7923540512570768)};
				\addplot[cb3] table [x = {eps}, y = {N150}] {\table};
				\addplot[mark=diamond*, mark options = cb3 ] coordinates{ (8, 0.8055031350824862)};

			%
				\addplot[domain = 0:10, cb4,  line width=0.5pt, dashed ] { (0.8489681621720339)};
				\addplot[cb4,dash pattern=on 0.5pt off 1pt,  line cap = round] table [x = {eps}, y = {LDP}] {\table};
				\legend{it = $50$,, it = $100$, ,it = $150$,,  max acc, local DP  }
			\end{axis}
		\end{tikzpicture}
		\caption{Model utility on Nhanes III versus $\epsilon_{\max}$ and training rounds. The dashed line shows peak accuracy on plaintext after 500 rounds; the diamond marks max convergence without noise. Dotted lines indicate LDP’s best accuracy at 500 rounds. Accuracy $=0$ means no convergence.}\label{fig:benchmarkNhanes}
	\end{minipage}

\end{figure}

Another trade-off comes from the number of iterations and the privacy budget.
The more rounds we train with a fixed privacy budget, the smaller the privacy budget per iteration and thus, the larger the noise.
On the other hand, less iterations may not be enough to reach a good convergence. \cref{fig:benchmarkNhanes} visualize this consideration based on NHANES III.
 Even without noise, a relatively small number of iterations may not be enough to reach the best possible parameter, but for higher number of iterations, it takes a larger $\epsilon$ to obtain a good utility.

 Hence, although the training of our logistic regression is relatively fast (see Table~\ref{tab:logreg}), there is still the drawback regarding the number of iterations and the privacy.
 However, this problem is a common problem using private GD and only indirectly dependent on our scheme.
 In fact, since we only add the noise at the end of the iteration, i.e., after all inputs of the clients have been combined, and not before, the amount of noise needed is reduced, as can be seen compared to LDP in \cref{fig:benchmarkDS} and \cref{fig:benchmarkNhanes}.
 This is because one individual does not influence the function outcome as heavily.

\section{Conclusion}\label{sec:conclusion}
In this work, we introduced the class of DyNMCFE, which supports DP and can be used in PPML settings where the data holders want to be excluded from the analysis and the authority has limited resources.
We gave a concrete instantiation of a DyNMCFE scheme, namely \lnmcfeot.
This scheme is faster than any previous instantiation of noisy FE, such that even the training of a logistic regression model can be realized in reasonable time.
We believe that it can also be embedded into other ML protocols to obtain output privacy.
To further improve the utility and efficiency, many optimizations known from privacy-preserving ML are possible, for example optimal partition of the privacy budget or batching.

\bibliographystyle{splncs04}

	\bibliography{local, cryptobib/crypto}

\begin{thebibliography}{10}
\providecommand{\url}[1]{\texttt{#1}}
\providecommand{\urlprefix}{URL }
\providecommand{\doi}[1]{https://doi.org/#1}

\bibitem{tensorflow2015}
Abadi, M., Agarwal, A., Barham, P., Brevdo, E., Chen, Z., Citro, C., Corrado, G.S., Davis, A., Dean, J., Devin, M., Ghemawat, S., Goodfellow, I., Harp, A., Irving, G., Isard, M., Jia, Y., Jozefowicz, R., Kaiser, L., Kudlur, M., Levenberg, J., Man\'{e}, D., Monga, R., Moore, S., Murray, D., Olah, C., Schuster, M., Shlens, J., Steiner, B., Sutskever, I., Talwar, K., Tucker, P., Vanhoucke, V., Vasudevan, V., Vi\'{e}gas, F., Vinyals, O., Warden, P., Wattenberg, M., Wicke, M., Yu, Y., Zheng, X.: {TensorFlow}: Large-scale machine learning on heterogeneous systems (2015), \url{https://www.tensorflow.org/}, software available from tensorflow.org

\bibitem{abadi2016deep}
Abadi, M., Chu, A., Goodfellow, I., McMahan, H.B., Mironov, I., Talwar, K., Zhang, L.: Deep learning with differential privacy. In: Proceedings of the 2016 ACM SIGSAC conference on computer and communications security. pp. 308--318 (2016)

\bibitem{PKC:ABKW19}
Abdalla, M., Benhamouda, F., Kohlweiss, M., Waldner, H.: Decentralizing inner-product functional encryption. pp. 128--157 (2019). \doi{10.1007/978-3-030-17259-6_5}

\bibitem{C:ACFGU18}
Abdalla, M., Catalano, D., Fiore, D., Gay, R., Ursu, B.: Multi-input functional encryption for inner products: Function-hiding realizations and constructions without pairings. pp. 597--627 (2018). \doi{10.1007/978-3-319-96884-1_20}

\bibitem{agrawal2022multi}
Agrawal, S., Goyal, R., Tomida, J.: Multi-input quadratic functional encryption: Stronger security, broader functionality. In: Theory of Cryptography Conference. pp. 711--740. Springer (2022)

\bibitem{C:AgrLibSte16}
Agrawal, S., Libert, B., Stehl{\'e}, D.: Fully secure functional encryption for inner products, from standard assumptions. pp. 333--362 (2016). \doi{10.1007/978-3-662-53015-3_12}

\bibitem{go-crypto-aes}
Authors, T.G.: crypto/aes package. \url{https://pkg.go.dev/crypto/aes@go1.23.2} (2024), \url{https://pkg.go.dev/crypto/aes@go1.23.2}, version go1.23.2

\bibitem{bakas2022heal}
Bakas, A., Michalas, A.: Heal the privacy: Functional encryption and privacy-preserving analytics. arXiv preprint arXiv:2205.03083  (2022)

\bibitem{bakas2022private}
Bakas, A., Michalas, A., Dimitriou, T.: Private lives matter: A differential private functional encryption scheme. In: Proceedings of the Twelveth ACM Conference on Data and Application Security and Privacy. pp. 300--311 (2022)

\bibitem{balle18a}
Balle, B., Wang, Y.X.: Improving the {G}aussian mechanism for differential privacy: Analytical calibration and optimal denoising. In: Dy, J., Krause, A. (eds.) Proceedings of the 35th International Conference on Machine Learning. Proceedings of Machine Learning Research, vol.~80, pp. 394--403. PMLR (10--15 Jul 2018), \url{https://proceedings.mlr.press/v80/balle18a.html}

\bibitem{baltico2017practical}
Baltico, C.E.Z., Catalano, D., Fiore, D., Gay, R.: Practical functional encryption for quadratic functions with applications to predicate encryption. In: Annual International Cryptology Conference. pp. 67--98. Springer (2017)

\bibitem{TCC:BoyChuPas14}
Boyle, E., Chung, K.M., Pass, R.: On extractability obfuscation. pp. 52--73 (2014). \doi{10.1007/978-3-642-54242-8_3}

\bibitem{carpov2020illuminating}
Carpov, S., Fontaine, C., Ligier, D., Sirdey, R.: Illuminating the dark or how to recover what should not be seen in fe-based classifiers. Proceedings on Privacy Enhancing Technologies  \textbf{2020}(2),  5--23 (2020)

\bibitem{chang2023privacy}
Chang, Y., Zhang, K., Gong, J., Qian, H.: Privacy-preserving federated learning via functional encryption, revisited. IEEE Transactions on Information Forensics and Security  \textbf{18},  1855--1869 (2023)

\bibitem{AC:CDGPP18}
Chotard, J., {Dufour Sans}, E., Gay, R., Phan, D.H., Pointcheval, D.: Decentralized multi-client functional encryption for inner product. pp. 703--732 (2018). \doi{10.1007/978-3-030-03329-3_24}

\bibitem{C:CDSGPP20}
Chotard, J., Dufour-Sans, E., Gay, R., Phan, D.H., Pointcheval, D.: Dynamic decentralized functional encryption. In: Micciancio, D., Ristenpart, T. (eds.) Advances in Cryptology -- CRYPTO 2020. pp. 747--775. Springer International Publishing, Cham (2020)

\bibitem{PKC:DatOkaTom18}
Datta, P., Okamoto, T., Tomida, J.: Full-hiding (unbounded) multi-input inner product functional encryption from the {$k$-L}inear assumption. pp. 245--277 (2018). \doi{10.1007/978-3-319-76581-5_9}

\bibitem{du2021dynamic}
Du, J., Li, S., Chen, X., Chen, S., Hong, M.: Dynamic differential-privacy preserving sgd. arXiv preprint arXiv:2111.00173  (2021)

\bibitem{dufour2018reading}
Dufour-Sans, E., Gay, R., Pointcheval, D.: Reading in the dark: Classifying encrypted digits with functional encryption. Cryptology ePrint Archive  (2018)

\bibitem{Dwork14}
Dwork, C., Roth, A.: The algorithmic foundations of differential privacy. Foundations and Trends{\textregistered} in Theoretical Computer Science  \textbf{9}(3--4),  211--407 (2014)

\bibitem{escobar2024computational}
Escobar, F.A., Canard, S., Laguillaumie, F., Phan, D.H.: Computational differential privacy for encrypted databases supporting linear queries. Proceedings on Privacy Enhancing Technologies  \textbf{4},  583–604 (2024)

\bibitem{fan2023fate}
Fan, T., Kang, Y., Ma, G., Chen, W., Wei, W., Fan, L., Yang, Q.: Fate-llm: A industrial grade federated learning framework for large language models. arXiv preprint arXiv:2310.10049  (2023)

\bibitem{garg2016candidate}
Garg, S., Gentry, C., Halevi, S., Raykova, M., Sahai, A., Waters, B.: Candidate indistinguishability obfuscation and functional encryption for all circuits. SIAM Journal on Computing  \textbf{45}(3),  882--929 (2016)

\bibitem{gentry2009fully}
Gentry, C.: Fully homomorphic encryption using ideal lattices. In: Proceedings of the forty-first annual ACM symposium on Theory of computing. pp. 169--178 (2009)

\bibitem{goldreich1998secure}
Goldreich, O.: Secure multi-party computation. Manuscript. Preliminary version  \textbf{78}(110),  1--108 (1998)

\bibitem{keller2020mp}
Keller, M.: Mp-spdz: A versatile framework for multi-party computation. In: Proceedings of the 2020 ACM SIGSAC conference on computer and communications security. pp. 1575--1590 (2020)

\bibitem{kim2018logistic}
Kim, A., Song, Y., Kim, M., Lee, K., Cheon, J.H.: Logistic regression model training based on the approximate homomorphic encryption. BMC medical genomics  \textbf{11},  23--31 (2018)

\bibitem{kim2019secure}
Kim, M., Lee, J., Ohno-Machado, L., Jiang, X.: Secure and differentially private logistic regression for horizontally distributed data. IEEE Transactions on Information Forensics and Security  \textbf{15},  695--710 (2019)

\bibitem{kim2018secure}
Kim, M., Song, Y., Wang, S., Xia, Y., Jiang, X., et~al.: Secure logistic regression based on homomorphic encryption: Design and evaluation. JMIR medical informatics  \textbf{6}(2),  e8805 (2018)

\bibitem{ploss1}
Li, M., Tian, Y., Zhang, J., Fan, D., Zhao, D.: The trade-off between privacy and utility in local differential privacy. In: 2021 International Conference on Networking and Network Applications (NaNA). pp. 373--378 (2021). \doi{10.1109/NaNA53684.2021.00071}

\bibitem{li2020federated}
Li, T., Sahu, A.K., Talwalkar, A., Smith, V.: Federated learning: Challenges, methods, and future directions. IEEE signal processing magazine  \textbf{37}(3),  50--60 (2020)

\bibitem{AC:LibTit19}
Libert, B., Titiu, R.: Multi-client functional encryption for linear functions in the standard model from {LWE}. pp. 520--551 (2019). \doi{10.1007/978-3-030-34618-8_18}

\bibitem{ligier2017privacy}
Ligier, D., Carpov, S., Fontaine, C., Sirdey, R.: Privacy preserving data classification using inner-product functional encryption. In: International Conference on Information Systems Security and Privacy. vol.~2, pp. 423--430. SciTePress (2017)

\bibitem{LBWstudy}
LogisticDx: Diagnostic tests for models with a binomial response. lbw: Low birth weight study data (8 2024), \url{https://rdrr.io/rforge/LogisticDx/man/lbw.html}

\bibitem{nhanes}
LogisticDx: Diagnostic tests for models with a binomial response. nhanes3: Nhanes iii data (8 2024), \url{https://rdrr.io/rforge/LogisticDx/man/nhanes3.html}

\bibitem{pcsStudy}
LogisticDx: Diagnostic tests for models with a binomial response. pcs: Prostate cancer study data (8 2024), \url{https://rdrr.io/rforge/LogisticDx/man/pcs.html}

\bibitem{uisStudy}
LogisticDx: Diagnostic tests for models with a binomial response. uis: Umaru impatct study data (8 2024), \url{https://rdrr.io/rforge/LogisticDx/man/uis.html}

\bibitem{ldp}
Mahawaga~Arachchige, P.C., Bertok, P., Khalil, I., Liu, D., Camtepe, S., Atiquzzaman, M.: Local differential privacy for deep learning. IEEE Internet of Things Journal  \textbf{7}(7),  5827--5842 (2020). \doi{10.1109/JIOT.2019.2952146}

\bibitem{PKC:MKMS22}
Mera, J.M.B., Karmakar, A., Marc, T., Soleimanian, A.: Efficient lattice-based inner-product functional encryption. pp. 163--193 (2022). \doi{10.1007/978-3-030-97131-1_6}

\bibitem{ploss2}
Naseri, M., Hayes, J., Cristofaro, E.D.: Toward robustness and privacy in federated learning: Experimenting with local and central differential privacy. CoRR  \textbf{abs/2009.03561} (2020), \url{https://arxiv.org/abs/2009.03561}

\bibitem{panzade2023fenet}
Panzade, P., Takabi, D.: Fenet: Privacy-preserving neural network training with functional encryption. In: Proceedings of the 9th ACM International Workshop on Security and Privacy Analytics. pp. 33--43 (2023)

\bibitem{rao2017aes}
Rao, S., Mahto, D., Yadav, D.K., Khan, D.: The aes-256 cryptosystem resists quantum attacks. Int. J. Adv. Res. Comput. Sci  \textbf{8}(3),  404--408 (2017)

\bibitem{ryffel2019partially}
Ryffel, T., Dufour-Sans, E., Gay, R., Bach, F., Pointcheval, D.: Partially encrypted machine learning using functional encryption. arXiv preprint arXiv:1905.10214  (2019)

\bibitem{eHR}
Slawomirski, L., et~al.: Progress on implementing and using electronic health record systems: Developments in oecd countries as of 2021. OECD Health Working Papers  \textbf{No. 160} (2023). \doi{10.1787/4f4ce846-en}

\bibitem{nist}
of~Standards, N.I., (NIST), T., Dworkin, M.J., Turan, M.S., Mouha, N.: Advanced encryption standard (aes) (2023-05-09 04:05:00 2023). \doi{https://doi.org/10.6028/NIST.FIPS.197-upd1}, \url{https://tsapps.nist.gov/publication/get_pdf.cfm?pub_id=936594}

\bibitem{googleDP}
Team, D.P.: Differential privacy library (dp lib v2.0.0) (2024), \url{https://github.com/google/differential-privacy}

\bibitem{viand2023verifiablefullyhomomorphicencryption}
Viand, A., Knabenhans, C., Hithnawi, A.: Verifiable fully homomorphic encryption (2023), \url{https://arxiv.org/abs/2301.07041}

\bibitem{waters2015punctured}
Waters, B.: A punctured programming approach to adaptively secure functional encryption. In: Annual Cryptology Conference. pp. 678--697. Springer (2015)

\bibitem{xu2019cryptonn}
Xu, R., Joshi, J.B., Li, C.: Cryptonn: Training neural networks over encrypted data. In: 2019 IEEE 39th International Conference on Distributed Computing Systems (ICDCS). pp. 1199--1209. IEEE (2019)

\bibitem{ZASH24}
Zalonis, J., Armknecht, F., Scheu-Hachtel, L.: Differentially private functional encryption. Proceedings on Privacy Enhancing Technologies  \textbf{2},  509–530 (2024)

\bibitem{zhao2019secure}
Zhao, C., Zhao, S., Zhao, M., Chen, Z., Gao, C.Z., Li, H., Tan, Y.a.: Secure multi-party computation: theory, practice and applications. Information Sciences  \textbf{476},  357--372 (2019)

\bibitem{ziller2021pysyft}
Ziller, A., Trask, A., Lopardo, A., Szymkow, B., Wagner, B., Bluemke, E., Nounahon, J.M., Passerat-Palmbach, J., Prakash, K., Rose, N., et~al.: Pysyft: A library for easy federated learning. Federated Learning Systems: Towards Next-Generation AI pp. 111--139 (2021)

\end{thebibliography}

\appendix
\section{Broader definition of NMCFE}\label{sec:extended-def}

\begin{definition}[Noisy Multi-Client Functional Encryption]
    Let $\{\cF_{n}^{m}\}_{n \in \N}$ be a family of sets $\cF_n$ of functions $f: \mathcal{X}_{1} \times \dots \times \mathcal{X}_{n} \to \mathcal{Y}$.
    Let $\Labels = \{0,1\}^* \cup \{\bot\}$ be a set of labels.
    A noisy multi-client functional encryption scheme (NMCFE) for $\cF_n$ and \Labels is a tuple of four efficient algorithms $\mathsf{NMCFE}=(\Setup, \Enc, \KeyGen, \Dec)$  of the following form:
    \begin{description}
        \item[$\Setup(1^\lambda, m, n)$:]
        Takes as input the security parameter $\lambda$, the number of clients $n$ and vector length $m$.
        It outputs a set of public parameters \pp, a master secret key \msk as well as secret keys \ski, for each slot $i \in [n]$.
        All of the remaining algorithms implicitly take \pp.
        \item[$\Enc(\ski, x_i, \ell)$:] Takes as input the secret key \ski for a slot $i \in [n]$, a message $x_i \in \mathcal{X}_i$ and a label $\ell \in \Labels$.
        It outputs a ciphertext $\ct_{i,\ell}$.
        \item[$\KeyGen(\msk, f, \ell, \distr)$:]  Takes as input the master secret key \msk, a function $f\in \cF_{n}$, a label $\ell \in \Labels$ and a distribution $\distr$ over some $\Delta \subseteq \mathcal{Y}$ such that $\Pr\left[f + \Delta \allowbreak \in \allowbreak \cF_{n}\right]=1$.
        Sample $\noise \leftarrow \distr$ and output a decryption key $\dk_{f,\ell}$.
        \item[$\Dec(\dk_{f, \ell}, \ct_{1,\ell}, \dots, \ct_{n,\ell})$:] Takes as input the decryption key $\dk_{f,\ell}$ and $n$ ciphertexts $\ct_{1,\ell}, \dots, \ct_{n,\ell}$, all encrypted under the same label used for the decryption key.
        It outputs a function evaluation $f(x_1, \dots, x_n) + \noise \in \mathcal{Y}$.
    \end{description}
    \end{definition}

The following is a generalized security definition that covers existing definitions for (noisy) MIFE and MCFE. Hence, our notation includes several options.
We will briefly explain some of these to enhance the clarity of the technical definition.
If we set the noise to always be zero, we obtain the classical security definition of MCFE.
Additionally, if we set $\Labels = \{\bot\}$, we get the security definition for MIFE.
Another consideration is the number of required encryption queries per slot $i \in [n]$, i.e., exactly one (\pos) or at least one (\posp).


The specification for handling the labels in the decryption keys is another consideration when building an NMCFE scheme. The keys can either be valid for all ciphertexts  (\all) or only for those, which share the same label as the key (\lab).
This is directly related to~\cref{eq:admissibility}, as the decryption must only protect the trivial attack, i.e., the inputs of the ciphertext queries posed by the adversary yield different function evaluations, if the key and the ciphertexts are generated under the same label.
If the key is not tied to the label (\all), then \cref{eq:admissibility} should hold for all queried functions, as long as only ciphertexts under the same label are combined.

Finally, it makes a difference at what point clients are corrupted, i.e., whether they are corrupted from the start (static  (\sta) corruptions) or at a later point (adaptive (\ad) corruptions).

In the following, $\distr_{\noise}$ defines the distribution that outputs \noise with probability one.

\begin{definition}[Security of NMCFE] \label{def:securitymcfe-extend}
	Let \nmcfe be an NMCFE scheme for $n$ clients and \Labels be a label set.
	For specifications $\textnormal{ww} \in \{\textnormal{\sta, ad}\},$ $\textnormal{xx} \in \{\textnormal{nh, mh}\},$
	$\textnormal{yy} \in \{\pos, \posp\},$
	$\textnormal{zz} \in \{\lab, \all\}$,
	and any security parameter $\lambda$,
	consider the following game \ww-\xx-\yy-\zz-$\mathsf{IND}_{\beta}^{\nmcfe}$ between an adversary \adv and a challenger \challenger.
	The game involves a set \HS of honest clients, initialized to $ \HS := [n] $, and a set of corrupted clients \CS, initialized to $\CS := \emptyset$.
	\begin{description}

		\item[Initialization:] At the beginning of their interaction, \challenger runs $(\pp, \msk, \{\sk_i\}_{i \in [n]}) \leftarrow \Setup(1^\lambda, m, n)$.
		Then, it chooses a random bit $\beta \leftarrow \bin$ and gives $\pp$ to \adv.
		\item[Encryption queries:] \adv can adaptively make encryption queries $\QEnc(i, x_i, \ell)$, which \challenger replies with $\ct_{i,\ell} \gets \Enc(\ski,  x_i, \ell)$.
		Denote by \qcenc the total number of queries of the form $\QEnc(i, \cdot, \ell)$.
		\item[Challenge encryption queries:] \adv can adaptively transmit challenge encryption queries  $\CQEnc(i, x_i^0, x_i^1, \ell)$, which are answered with $\ct_{i,\ell} \gets \Enc(\ski, x_i^\beta, \ell)$.
		Denote by \qccenc the total number of queries $\CQEnc(i, \cdot, \cdot, \ell)$.
		\item[Decryption key queries:] The adversary \adv can adaptively obtain functional decryption keys via queries $\QKeyGen(f, \ell, \noise)$. In response, \challenger returns $\dk_{f, \ell} \gets \KeyGen(\msk, f, \ell, \distr_{\noise})$.
		Denote by $\qkkey$ the total number of queries $\CQEnc(\cdot, \ell, \cdot, \cdot)$.
		\item[Challenge decryption key queries:] The adversary can adaptively pose queries $\CQKeyGen(f, \ell, \noise^0, \noise^1)$ to receive functional decryption keys.
		In response, \challenger returns $\dk_{f, \ell} \gets \KeyGen(\msk, f, \ell, \distr_{\noise^\beta})$.
		Denote by $\qkckey$ the total number of queries $\CQEnc(\cdot, \ell, \cdot, \cdot)$.
		\item[Corruption queries:] \adv can query the corruption oracle  $\QCor(i)$ to obtain $\ski$.
		\challenger set $\CS := \CS \cup \{i\}$ and $\HS := \HS \setminus \{i\}$.
		\item[Finalize:] \adv outputs a bit $\beta' \in \bin$.
		\challenger checks, if \adv acted admissibly.
		If not, \challenger sets $\beta' =0$.
		\adv wins, if $\beta' = \beta$.
	\end{description}

	Depending on the specification of the game defined in advance, we call \adv admissible, if all of the following conditions hold.
	Denote by $\qcil$ the total number of encryption and challenge encryption queries to slot $i$ under label $\ell$, i.e., $\qcil = \qcenc + \qccenc$.
	Similarly, let $\qkl$ be the total number of decryption and challenge decryption key queries under label $\ell$, i.e., $\qkl = \qkkey + \qkckey$.
	\begin{itemize}
		\item If $i \in \CS$, then for any query $\CQEnc(i, x_i^0, x_i^1, \ell)$, $x_i^0 = x_i^1$ .
		\item For any family storing all honest and all challenge encryption queries, i.e., $\{\CQEnc(i, x_i^0, x_i^1, \ell)~ \text{or } \QEnc(i, x_i, \ell)\}_{i \in \HS}$,
		for any family of inputs $\{x_i \in \mathcal{X}_{i}\}_{i \in \CS}$,
		any label $\ell \in \Labels$,
		any query $\CQKeyGen(f_{\ell}, \ell, \noise^0, \noise^1)$,
		we define $x_i^0:= x_i$ and $x_i^1 := x_i$ for $i \in \CS$
		and any slot queried to $\QEnc(i, x_i, \ell)$
		and $\noise^0 := \noise$ and $\noise^1 := \noise$ for any query $\QKeyGen(f_{\ell}, \ell, \noise)$, and we require that:
		\begin{align}\label{eq:admissibility}
			f_\ell(x_1^0, \dots, x_n^0) + \noise_\ell^0 = f_\ell(x_1^1, \dots, x_n^1) + \noise_\ell^1.
		\end{align}
		If one slot $i \in \HS$ has not been queried, then there is no restriction.
		\item If $\ww = \sta$, the adversary should output \CS at the beginning of the game and does not have access to $\QCor(\cdot)$ afterwards.
		\item If $\xx = \mh$, then  for all queries $\CQKeyGen(f, \ell, \noise^0, \noise^1)$, $\noise^0 = \noise^1=0$.
		\item If $\yy=\anyot$, then for any slot $ i \in [n]$ and $ \ell \in \Labels$, $\qcil \leq 1$.
		In other words, the adversary can pose at most one (challenge) encryption query per slot.
		\item If $\yy = \pos$, then for any slot $ i \in \HS$ and $ \ell \in \Labels$, it holds that if $\qccenc > 0$ or $\qkckey > 0$, then $\qcjl = 1$ for any slot $j \in \HS$.
		In other words, for any label, if there is one honest (challenge) encryption query or challenge decryption key query, then all honest slots have to be queried exactly once.
		\item If $\yy = \posp$, then for any slot $ i \in \HS$ and $ \ell \in \Labels$, it holds that if $\qccenc > 0$ or $\qkckey>0$, then $\qcjl > 0$ for any slot $j \in \HS$.
		In other words, for any label, the adversary makes either at least one challenge query for each slot $i \in \HS$ or none.
		Moreover, for any label, if the adversary queries a challenge decryption key, they also have to pose at least one (challenge) encryption query per slot.
		\item If $\zz = \all$, then for any $f \in \cF_n$, $\dk_{f,\ell} = \dk_{f, \ell'}$ for all $\ell, \ell' \in \Labels$.
		In other words, the decryption keys for a function $f$ are the same for each label and can thus be seen as independent of the label.

	\end{itemize}
	The advantage of $\adv$ in this game is defined as
	\begin{align*}
		\advantage^{\ww\text{-}\xx\text{-}\yy\text{-}\zz\text{-}\mathsf{IND}}_{\nmcfe, \adv}(\lambda)
		&= \left| \Pr(\ww\text{-}\xx\text{-}\yy\text{-}\zz\text{-}\mathsf{IND}_0^{\nmcfe}(\lambda, \adv) =1) \right.\\
		&-\left.\Pr(\ww\text{-}\xx\text{-}\yy\text{-}\zz\text{-}\mathsf{IND}_1^{\nmcfe}(\lambda, \adv) =1) \right|
	\end{align*}
	We say that an NMCFE scheme \nmcfe provides \ww-$\xx$-\yy-\zz-$\mathsf{IND}$-security, if
	$\advantage^{\ww\text{-}\xx\text{-}\yy\text{-}\zz\text{-}\mathsf{IND}}_{\nmcfe, \adv}(\lambda) \leq \mathsf{negl}(\lambda).$
\end{definition}

Note that this security definition is the same as \cref{def:securitymcfe} for $\ww\text{-}\xx\text{-}\yy\text{-}\zz=\sta\text{-}\nh\text{-}\pos\text{-}\lab$ and if all encryption keys are generated during setup.
In fact, for our proposed scheme, \lnmcfe, we will show that thee non-dynamic version fulfills $\sta\text{-}\nh\text{-}\pos\text{-}\lab\text{-}\mathsf{IND}$-security and conclude that this suffices for the \IND-security of \lnmcfe.

In case only one label is used, i.e., $\Labels= \{\bot\}$, the declaration of $\zz$ is omitted.
Moreover, we omit the label index in the notation above, e.g., only write $\qk$ instead of $q_{k,\ell}$.
In a similar manner, if no corruptions are allowed, we omit $\ww$.

To simplify our security proofs, we introduce the notion of 1-label security, as it is common in the MCFE setting~\cite{AC:LibTit19}.
In this modified game, the adversary is only allowed to pose challenge queries under a single label $\ell^*$.
Although this may seem as a strong restriction, we show via a sequence of hybrid arguments that is equivalent to Definition~\ref{def:securitymcfe}.

\begin{definition}[1-label security of NMCFE]
	Let \nmcfe be a NMCFE scheme for $n$ clients and $\Labels$ be a label set. For \ww, \xx, \yy, $\zz$  defined as in Definition~\ref{def:securitymcfe}
	and $\beta \in \{0,1\}$, we define the game
	$1$-\ww-\xx-\yy-\zz-$\mathsf{IND}_{\beta}^{\nmcfe}$ as in Definition~\ref{def:securitymcfe}, except for the following changes in the oracle queries:
	\begin{description}
		\item[Encryption queries:]
		\adv can adaptively make encryption queries $\QEnc(i, x_i, \ell)$, which \challenger replies with $\ct_{i,\ell} \gets \Enc(\ski, x_i, \ell)$.
		If one label corresponds to $\ell^*$ queried to \CQEnc, the game ends and returns 0.
		\item[Challenge encryption queries:]
		\adv can adaptively transmit challenge encryption queries  $\CQEnc(i, x_i^0, x_i^1, \ell)$, which are answered with $\ct_{i,\ell} \gets \Enc(\ski, x_i^\beta, \ell)$.
		These queries can be made on at most one label $\ell^*$.
		Further queries with distinct labels will be ignored.
		\item[Decryption key queries:]
		\adv can adaptively pose queries $\QKeyGen(f, \ell, \noise)$ to obtain functional decryption keys.
		\challenger returns $\dk_{f, \ell}  \gets \KeyGen(\msk, f, \ell, \distr_{\noise})$.
		If one label corresponds to $\ell'$ queried to \CQKeyGen, the game ends and returns 0.
		\item[Challenge decryption key queries:]
		The adversary can adaptively receive functional decryption keys by posing  queries $\CQKeyGen(f, \ell, \noise^0, \noise^1)$.
		In response, \challenger returns $\dk_{f, \ell} \gets \KeyGen(\msk, f, \ell, \distr_{\noise^\beta})$.
		These queries can be made on at most one label $\ell'$.
		Further queries with distinct labels will be ignored.

	\end{description}
	The same conditions defining the admissibility of an adversary as in Definition~\ref{def:securitymcfe}
	apply to an adversary \adv playing the $1$-\ww-\xx-\yy-\zz-$\mathsf{IND}_{\beta}^{\nmcfe}$ game.
	The advantage of $\adv$ is defined as
	\begin{align*}
		\advantage^{1\text{-}\ww\text{-}\xx\text{-}\yy\text{-}\zz\text{-}\mathsf{IND}}_{\nmcfe, \adv}(\lambda)
		 &= \left| \Pr(1\text{-}\ww\text{-}\xx\text{-}\yy\text{-}\zz\text{-}\mathsf{IND}_0^{\nmcfe}(\lambda, \adv) =1) \right.\\
		&-\left.\Pr(1\text{-}\ww\text{-}\xx\text{-}\yy\text{-}\zz\text{-}\mathsf{IND}_1^{\nmcfe}(\lambda, \adv) =1) \right|
	\end{align*}
	We say that an NMCFE scheme \nmcfe provides $1$-\ww-$\xx$-\yy-\zz-$\mathsf{IND}$-security, if
	$$\advantage^{1\text{-}\ww\text{-}\xx\text{-}\yy\text{-}\zz\text{-}\mathsf{IND}}_{\nmcfe, \adv}(\lambda) \leq \negl(\lambda).$$
\end{definition}

This definition allows $\ell^* \neq \ell'$.
However, note that if this is the case, then $\noise_l^0 =\noise_l^1$ for all decryption key queries $l \in [q_{k,\ell'}]$ under the condition that \adv is an admissible adversary.
Moreover, $f(x_1^0, \dots, x_n^0) = f(x_1^1, \dots, x_n^1)$ for all $\{x_i^b\}_{i\in [n], b \in \bin}$.
Thus, the adversary would always play the game in the $\xx =\mh$ setting, which is the weaker security definition.
The similarity of both security definitions in that particular setting has been proven for example by Libert et al.~\cite{AC:LibTit19}.
We will focus on the case where $\ell^*=\ell'$.

\begin{lemma} \label{lem:oneLab}
	Let \mcfe be a scheme that is $1\text{-}\ww\text{-}\xx\text{-}\yy\text{-}\zz\text{-}\mathsf{IND}$-secure. Then it is also secure against a \PPT adversary \adv against its $\ww\text{-}\xx\text{-}\yy\text{-}\zz\text{-}\mathsf{IND}$ security.
	Namely, for any \PPT adversary \adv, there exists a \PPT adversary $\cB$ such that.
	\begin{align*}
		\advantage^{\ww\text{-}\xx\text{-}\yy\text{-}\zz\text{-}\mathsf{IND}}_{\mcfe, \adv}(\lambda) \leq \ql \cdot \advantage^{1\text{-}\ww\text{-}\xx\text{-}\yy\text{-}\zz\text{-}\mathsf{IND}}_{\mcfe, \cB}(\lambda),
	\end{align*}
	where \ql denotes the number of distinct labels queried by \adv to \CQEnc and \CQKeyGen in the original security game.
\end{lemma}

\begin{proof}
	Let \adv be an efficient adversary in the $\ww\text{-}\xx\text{-}\yy\text{-}\zz\text{-}\mathsf{IND}$-security game.
	We show that \adv implies an MCFE adversary in the $1\text{-}\ww\text{-}\xx\text{-}\yy\text{-}\zz\text{-}\mathsf{IND}$.
	Assume \adv makes encryption and decryption key queries for $q_{\ell}$ distinct labels.
	Denote by $\ell_j, j \in [q_\ell]$ the $j^{\text{th}}$ distinct label that was queried to one of the challenge oracle.
	We use a standard hybrid argument over the distinct labels that \adv queries throughout the attack.
	More precisely, let $\game{k}$ with $k \in \{0, \dots, q_\ell\}$ be the game in which the challenger responds in the following way to encryption and decryption key queries $\QEnc(i, x_i^0, x_i^1, \ell_j)$ and $\QKeyGen(f, \ell_j, \noise^0, \noise^1)$:
	\begin{itemize}
		\item If $j \leq k$, reply with $\Enc(\ski, x_i^0, \ell_j)$ and $\KeyGen(f, \ell_j, \noise^0)$, respectively.
		\item If $j > k$, reply with $\Enc(\ski, x_i^1, \ell_j)$ and $\KeyGen(f, \ell_j, \noise^1)$, respectively.
	\end{itemize}
	We assume that all of the restrictions imposed by the admissibility of the adversary remain the same.
	By construction, an adversary \adv against the \ww-$\xx$-\yy-\zz-$\mathsf{IND}$-security of Definition~\ref{def:securitymcfe}
	yields a distinguisher between $\game{0}$ and $\game{q_\ell}$.
	On the other hand, for any $k \in \{0,1, \dots, q_\ell\}$, an efficient distinguisher $\adv_k$ between $\game{k}$ and $\game{k+1}$ implies the existence of an efficient adversary $\cB_k$ in the $1$-\ww-$\xx$-\yy-\zz-$\mathsf{IND}_\beta^{\mcfe}$ game such that $\advantage_{\adv_k}^{k,k+1} (\lambda)= \advantage_{\mcfe, \cB_k}^{1\text{-}\ww\text{-}\xx\text{-}\yy\text{-}\zz\text{-}\mathsf{IND}}(\lambda)$.
	The reason is that $\cB_k$ can simulate $\adv_k$'s view as follows.
	\begin{description}
		\item[Initialization:] $\cB_k$ obtains \pp from its own challenger and hands it to $\adv_k$.
		\item[Encryption queries:] For each query $\QEnc(i, x_i^0, x_i^1, \ell_j)$ from $\adv_k$, $\cB_k$ acts as follows.
			\begin{itemize}
				\item If $j \leq k$, $\cB_k$ poses the encryption query $\QEnc(i, x_i^0, \ell_j)$ to its own oracle and forwards the response to $\adv_k$.
				\item If $j = k+1$,  $\cB_k$ poses the challenge encryption query $\CQEnc(i, x_i^0, x_i^1, \ell_j)$ to its own challenger and relays the response back to $\adv_k$.
				\item Otherwise, $\cB_k$ transmits the query $\QEnc(i, x_i^1, \ell_j)$ and returns the answer to $\adv_k$.
			\end{itemize}
		\item[Decryption key queries:] $\cB_k$ does the following for each decryption key query $\CQKeyGen(f, \ell_j, \noise^0, \noise^1)$ from $\adv_k$.
			\begin{itemize}
				\item If $j \leq k$, $\cB_k$ sends the query $\QKeyGen(f, \ell_j, \noise^0)$ to its own challenger and relays the answer back to $\adv_k$.
				\item If $j = k+1$,  $\cB_k$ poses the challenge encryption query $\CQKeyGen(f, \ell_j, \noise^0, \noise^1)$  and sends the response to $\adv_k$.
				\item Otherwise, $\cB_k$ queries $\QKeyGen(f, \ell_j, \noise^1)$ and transmits the answer to $\adv_k$.
			\end{itemize}
		\item[Corruption queries:] For each query $\QCor(i)$ from $\adv_k$, $\cB_k$ makes the same query to its own corruption oracle and transmits the answer to $\adv_k$.
		\item[Finalize:] $\cB_k$ outputs the same bit $\beta'$ as $\adv_k$.
	\end{description}
	Summing up the underlying probabilities concludes the proof.
\end{proof}

	\section{Security proofs}
Throughout this section, denote by \Win{i} the probability that an adversary \adv outputs the correct $\beta$ in \game{i}, i.e., that \adv wins \game{i}.
\subsection{Security of \nmifeot}\label{sec:security-nmife}
To prove the security of \nmifeot, we need the information-theoretical security of \mifeot.
Abdalla et al.\cite{PKC:ABKW19} showed:

\begin{theorem}\label{thm:secmife}(\cite{PKC:ABKW19})
The \mifeot scheme as described above fulfills \ad-\mh-\pos-$\mathsf{IND}$-security.
In other words, for any adversary \adv, $\advantage_{\mifeot, \adv}^{\ad\text{-}\nh\text{-}\pos\text{-}\mathsf{IND}}(\lambda) = 0$.
\end{theorem}
Using this theorem, we can now show the security of \nmifeot.
\begin{proof}[Proof of Theorem~\ref{thm:sec-nmifeot}]
	Consider the adversary \adv against the $\ad\text{-}\nh\text{-}\pos\text{-}\mathsf{IND}$-security of \nmifeot.
	We prove the theorem using two reductions.
	Let \sel-\sta-\nh-\pos-$\mathsf{IND}_{\beta}^{\mifeot}(\lambda)$ be a variant of the \sta-\nh-\pos-$\mathsf{IND}_{\beta}^{\mifeot}(\lambda)$ game, where the selective adversary has to additionally specify the encryption challenges $\{\xb_i^b\}_{i \in [n], b \in \{0,1\}}$ together with the corrupted set \CS at the beginning of its interaction with the challenger.
	First, consider an adversary \(\cB_1\) against the \sel-$\sta\text{-}\nh\text{-}\pos\text{-}\mathsf{IND}$ security of \nmifeot.
	\(\cB_1\) can simulate \adv's view in the following manner.

		Upon interaction with \adv, \(\cB_1\) guesses \adv's ciphertext queries \((\xb_i^0, \xb_i^1)\) as well as the set of corrupted clients \CS and commits them along with all key queries \((f_{\yb }^l, \noise_l^0, \noise_l^1)\), which are admissible with respect to \(\{(\xb_i^0, \xb_i^1)_{i\in[n]}\}\)  to their own selective oracle.

		Whenever \adv poses a ciphertext query \((\xb_i^0, \xb_i^1)\), \(\cB_1\) checks if their guess was correct. If so, they query their own oracle and send the result to \adv. Otherwise, they output 0.

		Whenever \adv poses a corruption query for $i \in [n]$, $\cB_1$ checks if $i \in \CS$. If so, they give $\sk_i$ to \adv. Otherwise, they output 0.

		Whenever \adv poses a key query \((f_{\yb}^l, \noise_l^0, \noise_l^1)\), \(\cB_1\) checks if it belongs to its set of key queries. If so, they query their own oracle and send the result to \adv. Otherwise, they output 0.

		If all guesses were correct, $\cB_1$ outputs the same bit as \adv.

	Note that if all queries were guessed successfully in advance, \(\cB_1\) perfectly simulates \adv's view.
	The probability that all ciphertext queries are guessed correctly is exactly \(q^{-2nm}\).
	On the other hand, the probability that $\cB_1$ guesses \CS correctly is $2^{-n}$.
	If \(\cB_1\)'s guess coincides with \adv's initial, then all of the key queries \adv poses are covered by the set of key queries \(\cB_1\) committed to their oracle as it contains all admissible key queries with respect to the same ciphertexts. Thus, it holds that
	\begin{align*}
		\advantage_{\nmifeot, \adv}^{\ad\text{-}\nh\text{-}\pos\text{-}\mathsf{IND}}(\lambda) \leq 2^{-n} q^{-2nm} \advantage_{\nmifeot, \cB_1}^{\sta\text{-}\nh\text{-}\pos\text{-}\sel\text{-}\mathsf{IND}}(\lambda).
	\end{align*}
	Next, let \(\cB_2\) be an adversary against the $\ad\text{-}\mh\text{-}\pos\text{-}\mathsf{IND}$-security of \mifeot for the family \(\cF_{q,n+\lceil \qk/m \rceil}\), where \qk is the amount of key queries \(\cB_1\) poses.
	They simulate the view of \(\cB_1\) as follows.

		 \(\cB_2\) obtains the set of ciphertext queries \(\{\xb_i^b\}_{i \in [n], b \in \{0,1\}}\) and function queries \(\{(f_{\yb}^l, \noise_l^0, \noise_l^1)\}_{l\in [\qk]}\) from \(\cB_1\).
		For \(j \in [\lceil \qk/m \rceil]\), \(\cB_2\) sets \(\xb_{n+i}^{\beta} = (\noise_{1+(i-1)m}^{\beta}, \dots,\allowbreak \noise_{m+(i-1)m}^{\beta})\), where \(\noise_{l}^{\beta}=0\) for all \(l > \qk\).
		 $\cB_2$ also receives the set of corrupted clients, \CS. They query their own corruption oracle on any index $i \in \CS$ to obtain \ski.
		$\{\sk_i\}_{i \in \CS}$ is handed to $\cB_1$.

		 \(\cB_2\) answers \(\cB_1\)'s encryption queries for \(i \in [n]\) by invoking their own oracle to obtain \(\ct_i\). This is directly sent to \(\cB_1\).
		Further, \(\cB_2\) also queries their oracle on \((\xb_i^0, \xb_i^1)\) for \(i \in [\lceil \qk/m \rceil]\).

		 To answer the \(l^{th}\) decryption key query of \((f_{\yb}^l, \noise_l^0, \noise_l^1)\) of $\cB_1$, \(\cB_2\) sets $\widetilde{\yb} = (\yb^\top, \mathbf{e}_l^\top)^\top$, where $\mathbf{e}_l$ is the vector of unity of dimension \(dm\) with \(d \in \N\) being the smallest number such that \(\qk \leq dm\).
		Parsing $\widetilde{\yb} = (\widetilde{\yb}_1, \dots, \widetilde{\yb}_{n+[\lceil \qk/m \rceil]}) \in \Z^{m(n+[\lceil \qk/m \rceil])}$, they query their oracle for $f_{\widetilde{\yb}}^l$
		to receive \(z\).
		Finally, $\cB_2$ calculates \(\widetilde{z} = z - \sum_{i \in [\lceil \qk/m \rceil]} \langle \ct_{n+i}, \widetilde{\yb}_{n+i}\rangle\) and sends it to \(\cB_1\).

		At the end of their interaction, \(\cB_2\) outputs the same bit as \(\cB_1\).

	Clearly, the ciphertext queries are equivalent in both cases. \(\cB_2\)'s key queries are admissible with respect to~(\ref{eq:admissibility}) as by construction,
	\begin{align*}
		\sum_{i \in [n+\lceil \qk/m \rceil]} \langle\ct_i, \widetilde{\yb}^l_{i}\rangle -z = \sum_{i \in [n]} \langle \xb_i^\beta, \yb^l_i \rangle + \noise_{l}^\beta
		= \sum_{i \in [n]} \langle \xb_i^0, \yb^l_i \rangle + \noise_{l}^0
	\end{align*}
	for all \(\beta \in \{0,1\}\) by the restrictions imposed on \(\cB_1\).

	Moreover, it holds that
	\begin{align*}
		\widetilde{z} &=\sum_{i \in [n+\lceil \qk/m \rceil]} \langle\sk_i, \widetilde{\yb}^l_{i}\rangle - \sum_{i \in [\lceil \qk/m \rceil]} \langle \xb_{n+i} + \sk_{n+i}, \widetilde{\yb}^l_{n+i}\rangle\\
		&= \sum_{i \in [n]} \langle\sk_i, \yb^l_{i}\rangle - \noise_l^\beta.
	\end{align*}
	Thus, \(\cB_2\) perfectly simulates \(\cB_1\)'s view.
	But as \(\advantage_{\mife^{ot}, \cB_2}^{\ad\text{-}\mh\text{-}\pos\text{-}\mathsf{IND}}(\lambda) = 0\) by Theorem~\ref{thm:secmife}, the claim follows.
\end{proof}

\subsection{Security of \lnmcfeot}\label{sec:sec-lnmcfeot}

First, let us recall the formal definition of a PRF.

\begin{definition}\label{def:PRF}
	For any PRF \PRF from $\D$ to $\cR$ and any security parameter $\lambda$ we define the experiment $\mathsf{IND}_\beta^{\PRF}$ between an adversary \adv and a challenger \challenger in the following way.
	\begin{description}
		\item[Initialization:] \challenger samples $\mathsf{K} \sampler \{0,1\}^\lambda$ and $\beta \sampler \bin$.
		\item[Challenge phase:] For a query $\ell \in \D$ posed by \adv, \challenger returns $\PRF(\mathsf{K}, \ell)$ if $\beta = 0$ and $\RF(\ell)$ otherwise.
		Here, $\RF$ denotes a function computed on the fly.
		\item[Finalize:] \adv outputs a bit $\beta' \in \bin$.
		If $\beta' = \beta$, \adv wins.
	\end{description}
	The advantage of an adversary \adv is defined as
	\begin{align*}
		\advantage_{\PRF, \adv}(\lambda) =\left| \Pr(\mathsf{IND}_0^{\PRF}(\lambda, \adv) =1)\right.
		- \left.\Pr(\mathsf{IND}_1^{\PRF}(\lambda, \adv) =1) \right|.
	\end{align*}
	\PRF is called secure, if for any \PPT adversary \adv, there exists a negligible function $\negl$ such that $\advantage_{\PRF, \adv}(\lambda) \leq \negl(\lambda)$.
\end{definition}

For our security proof, we need the non-dynamic version of \lnmcfeot, which is defined as follows.

\begin{construction}[Non-dynamic \lnmcfeot] \label{def:sta-lnmcfeot}
	Let \(\cF_{q,n}^m\) be the class of multi-input inner products over \(\Z_q\). The non-dynamic \lnmcfeot scheme for \(\cF_{q,n}^m\) consists of the following algorithms:
	\begin{description}
		\item[\(\Setup(1^\lambda, m, n)\):]
		On input of the security parameter $\lambda$, number of clients $n$ and vector length $m$,
		sample \(\ski\sampler \{0,1\}^\lambda\)  for all slots \(i \in [n]\).
		Set \(\msk := (\{\sk_i\}_{i \in [n]})\) and output both as well as the public parameters $\pp$.
		\item[\(\Enc(\sk_i, \xb_i, \ell)\):]
		On input of the secret key $\ski$, a plaintext $\xb_i \in \Z_q^m$  for slot $i \in [n]$ and a label $\ell \in \Labels$,
		calculate $\zeta = \PRF(\ski, \ell)$ and return \(\ct_{i,\ell} = (i, \ell, \mat c_i := \xb_i + \zeta \mod q)\).

		\item[\(\KeyGen(\msk, f_{\yb}, \ell, \distr)\):]
		On input of the master secret key \msk, a function $f_{\yb}$, a label $\ell \in \Labels$ and a noise distribution $\distr$,
		sample \(\noise \leftarrow \distr\) and calculate $\zeta_i = \PRF(\ski, \ell)$ for all \(i \in [n]\).
		Return \(\dk_{\yb,\ell} =(\ell, \yb, z := \sum_{i \in [n]} \langle \zeta_i, \yb_i \rangle - \noise \mod q)\).
		\item[\(\Dec(\dk_{\yb, \ell}, \{\ct_{i,\ell}\}_{i \in [n]})\):]
		On input of the decryption key \(\dk_{\yb, \ell}\) and ciphertexts \(\ct_{i,\ell}= (i, \ell, \mat c_i)\), parse \(\dk_{\yb, \ell} =(\ell, \{\yb_i\}_{i \in [n]}, z)\) and return \(\sum_{i \in [n]} \langle \mat c_i, \yb_i \rangle -z \mod q\).
	\end{description}
\end{construction}

\begin{theorem}\label{thm:nmcfe-dyno}
	Let \nmcfe be as in Construction \ref{def:sta-lnmcfeot} and \lnmcfe as in Construction \ref{def:lnmcfeot}.
	Then, for any \ppt adversary \adv against the \IND-security of \lnmcfe, there exists a \ppt adversary $\cB$ such that
	\begin{align*}
		\advantage^{\mathsf{IND}}_{\lnmcfeot, \adv}(\lambda) \leq \advantage^{\sta\text{-}\nh\text{-}\pos\text{-}\lab\text{-}\IND}_{\nmcfe, \cB}(\lambda).
	\end{align*}
\end{theorem}

\begin{proof}
	We denote by $\mat 0_k$ the all-zero vector of length $k$.
	$\cB$ simulates \adv's view as follows.
	Upon interaction, $\cB$ sets $n := \nmax$ and $m := \mmax$, which works as both are at most polynomial in $\lambda$.
	Moreover, $\cB$ sets $\cE := \emptyset$ as the set of indices for which an encryption key has been queried.

	If \adv poses a corruption query for slot $i \in [n]$, $\cB$ queries their own corruption oracle to obtain $\sk_i$, which is forwarded to \adv.
	$\cB$ updates $\cE := \cE \cup \{i\}$.

	Whenever $\adv$ queries its encryption key oracle on index $i$, $\cB$ sets $\cE := \cE \cup \{i\}$.

	Whenever \adv poses an encryption query $(i, \xb_i^0, \xb_i^1, \ell)$, where $\xb_i^\beta$ is of length $m_{\ell}$ $\cB$ checks if $i \in \cE$.
	If so, they query there own oracle $\CQEnc(i, \widetilde{\xb_i}^0, \widetilde{\xb_i}^1, \ell)$, where $\widetilde{\xb_i}^{\beta} = (\xb_i^\beta, \mat 0_{m-m_{\ell}}) \in \Z_q^{m}$ to obtain $\ctil$.
	They hand $\ctil' = (\ctil[1], \dots, \ctil[m_{\ell}])$ to \adv.
	Otherwise, they return nothing.

	Whenever \adv poses a decryption key query $(S, f, \ell, \noise^0, \noise^1)$, $\cB$ checks if $S \subseteq \cE$. If not, return $\bot$. Otherwise, $\cB$ interprets $f = \{\yb_i\}_{i \in S}$.
	They set $\widetilde{\yb_i} = \yb_i$ for $i \in S$, $\widetilde{\yb_i} = \mat 0_{m_{\ell}}$ for $i \in [n]\setminus S$ and interpret $f' = \{\widetilde{\yb_i}\}_{i \in [n]}$.
	$\cB$ queries their own oracle on $(f', \ell, \noise^0, \noise^1)$ and forwards the result.

	In the end, $\cB$ outputs the same bit as $\adv$.
	This perfectly simulates \adv's view.
	Note that $\cB$ is admissible by the admissibility of $\adv$.
\end{proof}

Now, in order to prove the security of Construction \ref{def:lnmcfeot}, i.e., \cref{thm:sec-nmifeot}, we simply have to prove the following theorem.

\begin{theorem}
	If \nmifeot from Construction~\ref{def:nmifeot} is \sta-\nh-\pos-$\mathsf{IND}$ secure and \PRF is secure, then \nmcfe from Construction \ref{def:sta-lnmcfeot} is \sta-\nh-\pos-\lab-$\mathsf{IND}$-secure. In particular, for any $\PPT$ adversary \adv,
 there exist $\PPT$ adversaries $\cB_1$ and $\cB_2$ such that
\begin{align*}
	\advantage^{\sta\text{-}\nh\text{-}\pos\text{-}\lab\text{-}\mathsf{IND}}_{\nmcfe, \adv}(\lambda) &\leq 2 n \ql \cdot \advantage_{\PRF, \cB_1}(\lambda)\\
	&+ \ql \cdot \advantage^{\sta\text{-}\nh\text{-}\pos\text{-}\mathsf{IND}}_{\nmifeot, \cB_2}(\lambda),
\end{align*}
where \ql denotes the number of distinct labels queried to \CQEnc and \CQKeyGen.
\end{theorem}

\begin{proof}
For simplicity, we consider the case where \adv only queries \CQEnc and \CQKeyGen on label $\ell^*$ and never on any other label.
We show that for $\PPT$ adversaries $\cB$ and $\cB'$, it holds that
\begin{align*}
	\advantage^{1\text{-}\sta\text{-}\nh\text{-}\pos\text{-}\lab\text{-}\mathsf{IND}}_{\nmcfe, \adv}(\lambda)
	&\leq 2n \cdot \advantage_{\PRF, \cB}(\lambda)\\
	&+ \advantage^{\sta\text{-}\nh\text{-}\pos\text{-}\mathsf{IND}}_{\nmifeot, \cB'}(\lambda).
\end{align*}
By Lemma~\ref{lem:oneLab}, this concludes the theorem.

Assume that there are $h$ honest clients, which are by definition of the security game known by the challenger before \adv can make any queries and let $\HS = \{i_1, \dots, i_h\}$.
For our proof, we utilize a sequence of hybrid games.
Denote by \Win{i} the probability that \adv outputs the correct $\beta$ in game \game{i}, i.e., that \adv wins \game{i}.
The games are as follows:
\begin{description}
	\item[$\game{0}$:] This game corresponds to the case  $\beta =0$.
	\item[\game{1}:] This game is the same as \game{0}, except to answer any encryption or decryption key queries for label $\ell \in \Labels$, \challenger sets $\zeta_{i_\eta} = \RF(i_{\eta}, \ell)$ for all $\eta \in [h]$ with \RF being a random function computed on the fly.
	In other words, the pseudorandom values of the honest clients are changed to truly random values.
	In Lemma \ref{lem:g0g1}, we exhibit a $\PPT$ adversary $\cB_0$ such that:
	\begin{align*}
		|\Win{0}-\Win{1}| \leq h \cdot \advantage_{\PRF, \cB_0}(\lambda).
	\end{align*}
	\item[\game{2}:] This game is the same as \game{1}, except that \challenger answers all challenge queries with $\beta =1$, i.e.,
		$\ct_{i, \ell} = \zeta_i + \xb_i^1$
	as the output of our encryption for all $i \in [n], \ell \in \Labels$ and
	$z = \sum_{i \in [n]} \langle \zeta_i, \yb_i \rangle - \noise^1$
	as the output of the key generation algorithm.
	In Lemma \ref{lem:g1-g2}, we exhibit a $\PPT$ adversary $\cB_1$ such that:
	\begin{align*}
		|\Win{0}-\Win{1}| \leq \advantage^{\sta\text{-}\nh\text{-}\pos\text{-}\mathsf{IND}}_{\nmifeot, \cB_1}(\lambda).
	\end{align*}
	\item[\game{3}:] This game corresponds to the case where \challenger sets $\beta=1$ in the original game, i.e., it is the same as \game{1} except that $\zeta_i$ are  generated by \PRF instead of \RF for all $i \in [n]$.
	By the same argument as in Lemma \ref{lem:g0g1}, it holds that there exists a $\PPT$ adversary $\cB_2$ such that:
	\begin{align*}
		|\Win{2}-\Win{3}| \leq h \cdot \advantage_{\PRF, \cB_2}(\lambda).
	\end{align*}
\end{description}
Putting everything together with the fact that there can be at most $n$ honest clients, we obtain the theorem.
\end{proof}

\begin{lemma}\label{lem:g0g1}
	There exists a $\PPT$ adversary $\cB_0$ such that
	$$|\Win{0}-\Win{1}| \leq h \cdot \advantage_{\PRF, \cB_0}(\lambda),$$
	where $h$ denotes the number of honest clients.
\end{lemma}
\begin{proof}
	We prove the lemma by using a hybrid argument over the $h$ honest clients, relying on the security of the \PRF.
	Let $\HS = \{i_1, \dots, i_h\}$.
	Without loss of generality, we assume  $i_1 < \dots < i_h$.
	Define \game{0, \eta}, $\eta \in [h] $, to be the same game as  $\game{0, \eta-1}$ where $\game{0,0} = \game{0}$, except to answer any encryption or decryption key queries for label $\ell \in \Labels$, \challenger sets $\zeta_{i_\eta} = \RF(i_{\eta}, \ell)$ with \RF being a random function computed on the fly.

	It is clear that
	\begin{align}\label{eq:hybridEta}
		|\Win{0}&-\Win{1}|
		= \sum_{\eta \in [h]} |\Win{0, \eta-1} - \Win{0, \eta}|
	\end{align}
	 as $\game{0,0} = \game{0}$ and $\game{1} = \game{0, h}$.

	We show that for every $\eta \in [h]$, there exists $\PPT$ adversary $\cB_{0, \eta}$ against the security of the \PRF such that
	\begin{align*}
		|\Win{0, \eta-1} - \Win{0, \eta}| \leq \advantage_{\PRF, \cB_{0, \eta}}(\lambda).
	\end{align*}

	We build $\cB_{0, \eta}$ so that it simulates \game{0, \eta -1 + \beta} to \adv when interacting with the experiment as described in Definition~\ref{def:PRF}.
	Given \CS sent by \adv, run $(\pp, \msk =  \{\ski\}_{i\in [n]}) \leftarrow \Setup(1^\lambda, m,n)$ and give $\pp$ and $\{\ski\}_{i \in \CS}$ to \adv.
	For all $i \in \CS$, $\cB_{0, \eta}$ can answer queries $\QEnc(i, \xb_i, \ell)$ and $\CQEnc(i, \xb_i^0, \xb_i^1, \ell^*)$ using $\sk_i$ to calculate $\ctil=\Enc(\ski, \xb_i, \ell)$ and $\ctil=\Enc(\ski, \xb_i^0, \ell^*)$, respectively.
	The ciphertext $\ctil$ is returned to \adv.

	In order to answer (challenge) encryption queries for $i \in \HS$ or (challenge) decryption key queries, $\cB_{0, \eta}$ sets $\zeta_i$ in the following manner:
	\begin{itemize}
		\item $i < i_\eta$: $\zeta_i = \RF(i, \ell)$,
		\item $i = i_\eta$: Query own oracle on input $(i_\eta, \ell)$ to obtain $\zeta_i = \zeta_{i_\eta}$,
		\item $i > i_\eta$: $\zeta_i = \PRF(\ski, \ell)$.
	\end{itemize}
	Any $\QEnc(i, \xb_i, \ell)$ and $\CQEnc(i, \xb_i^0, \xb_i^1, \ell^*)$ queries are answered using this $\zeta_i$, i.e., $\cB_{0, \eta}$ returns $\ctil= (i, \ell, \xb_i + \zeta_i \mod q$) and $\ctil= (i, \ell, \xb_i^0 + \zeta_i \mod q)$, respectively.
	To answer queries to $\QKeyGen(f_{\yb}, \noise, \ell)$ and $\CQKeyGen(f_{\yb}, \noise^0, \noise^1, \ell^*)$, first, $\cB_{0, \eta}$ calculates $\zeta_i = \PRF(\ski, \ell)$ for all $i \in \CS$.
	Then, for all $i \in \HS$, $\cB_{0, \eta}$ calculates $\zeta_i$ as above.
	Identify $f_{\yb}$ as $\yb = (\yb_1, \dots, \yb_n)$ and compute $z= \sum_{i \in [n]} \langle \zeta_i, \yb_i \rangle - \noise$, where $\noise:= \noise^0$ in challenge decryption key queries.
	Finally, send $\dk_{\yb,\ell} = (\ell, \yb, z)$ to \adv.

	In the end, $\cB_{0,\eta}$ outputs the same bit as \adv.
	Note that $\cB_{0,\eta}$ perfectly simulates \adv's view and is admissible with respect to $\adv$'s restrictions.
	Thus, by \cref{eq:hybridEta}, the claim follows.
\end{proof}
\begin{lemma}\label{lem:g1-g2}
	There exists a $\PPT$ adversary $\cB_1$ such that
	\begin{align*}
		|\Win{0}-\Win{1}| \leq \advantage^{\sta\text{-}\nh\text{-}\pos\text{-}\mathsf{IND}}_{\nmifeot, \cB_1}(\lambda).
	\end{align*}
\end{lemma}
\begin{proof}
	We build an adversary $\cB_1$ against the $\sta\text{-}\nh\text{-}\pos\text{-}\mathsf{IND}$ security of $\nmifeot$ as follows.

	Given \CS sent by \adv, $\cB_1$ samples $\mathsf{K}_i \leftarrow \{0,1\}^\lambda$ and hands $\{\ski\}_{i \in \CS}$ to \adv.
	Any encryption queries $\QEnc(i, \xb_i, \ell)$ from \adv are answered with $\ctil = (i, \ell, \xb_i + \zeta_i \mod q)$, where $\zeta_i = \PRF(\sk_i, \ell)$ if $i \in \CS$ and $\zeta_i = \RF(i, \ell)$ else, where \RF denotes a random function.
	Similarly, all decryption key queries $\QKeyGen(f_{\yb}, \noise, \ell)$ are answered by first calculating $\zeta_i = \PRF(\sk_i, \ell)$ if $i \in \CS$ and $\zeta_i = \RF(i, \ell)$ otherwise.
	Then, $\cB_1$ sets $z= \sum_{i \in [n]} \langle \zeta_i, \yb_i \rangle - \noise$ by identifying $f_{\yb}$ as $\yb = (\yb_1, \dots, \yb_n)$ and returns $\dk_{i, \ell}$ to \adv.

	Whenever $\cB_1$ obtains a  query $\CQEnc(i, \xb_i^0, \xb_i^1, \ell^*)$, they query their own encryption oracle on input $(i, \xb_i^0, \xb_i^1)$.
	If $i \in \HS$, the received ciphertext is directly forwarded to \adv.
	Otherwise, $\cB_1$ ignores the output of its oracle and calculates $\mat c_i = \PRF(\ski, \ell^*) + \xb_i^0 \mod q$ instead.
	The ciphertext $\ct_{i, \ell^*} = (i, \ell^*, \mat c_i)$ is then given to \adv.
	Note that in this case, it has to hold that $\xb_i^0 = \xb_i^1$ in order for \adv to be admissible.
	Thus,  $\cB_1$ does not make any assumptions about the value of $\beta$ by choosing to encrypt $x_i^0$.

	Upon receiving a challenge decryption key query $\CQKeyGen(f_{\yb}, \noise^0, \noise^1, \ell^*)$, $\cB_1$  views  $f_{\yb}$ as $\yb = (\yb_1, \dots, \yb_n)$ and calculates $z_1 = \sum_{i \in \CS} \langle \zeta_i, \yb_i \rangle$.
	Then, they queries their own decryption key oracle on input $(f_{\widetilde{\yb}}, \noise^0, \noise^1)$, where $\widetilde{\yb}_i = \yb_i$ if $i \in \HS$ and $0$ else, to obtain $\dk_{\widetilde{\yb}} = (\widetilde{\yb}, z_2)$.
	Send $\dk_{\yb, \ell^*} := (\ell^*, \yb, z_1+z_2)$ to \adv.

	At the end of their interaction,  $\cB_1$ outputs the same bit as \adv.
	If \adv is admissible, then so is $\cB_1$.
	First, $\cB_1$ queries each slot $i \in [n]$ of \nmifeot exactly once, if \adv poses exactly one challenge query $\CQEnc(i, \xb_i^0, \xb_i^1, \ell^*)$ and uses only one label $\ell^*$.
	Second, note that for $\cB_1$'s decryption key queries, it holds that
	\begin{align*}
		&\sum_{i \in [n]} \langle \xb_i^0, \widetilde{\yb}_i \rangle + \noise^0 = \sum_{i \in \HS} \langle \xb_i^0, \yb_i \rangle + \noise^0\\
		=&\sum_{i \in \HS} \langle \xb_i^1, \yb_i \rangle + \noise^1
		= \sum_{i \in [n]} \langle \xb_i^1, \widetilde{\yb}_i \rangle + \noise^1
	\end{align*}
	because $\sum_{i \in \CS} \langle \xb_i^0, \yb_i \rangle = \sum_{i \in \CS} \langle \xb_i^1, \yb_i \rangle$ by admissibility of \adv.

	Further, $\cB_1$ perfectly simulates \adv's view.
	This is clear for any (challenge) encryption and regular decryption key queries.
	In case of challenge decryption key queries, note that
	\begin{align*}
		z_1+z_2 = \sum_{i \in \CS} \langle \PRF(\ski, \ell^*), \yb_i \rangle + \sum_{i \in \HS} \langle \zeta_i, \yb_i \rangle - \noise^\beta,
	\end{align*}
	where $\zeta_i$ is a random value.
Thus, the lemma follows.
\end{proof}
	\section{Derivation of Formulas}\label{ap:logreg}


To implement the logistic regression, we have to know what the coefficients of $\yb$ are for each element in $\xb_i$.
For brevity, let  $\xb_i[m+1] = y_i$.
Writing~\cref{eq:theta-j} for $j \in [0;m]$ as a linear function for fixed coefficients up to degree 4 yields
\begin{align*}
    \theta&[j] := \theta[j] + \sum_{i \in [n]} \frac{\alpha}{n}\left(\left(- \frac{1}{2} +(a_1 \theta[0]^3 -a_2 \theta[0]) + y_i \right)\xb_i[j]\right.\\
    &+\sum_{k \in [m]}(-a_2 \theta[k] + 3a_1 \theta[0]^2)\xb_i[k] \xb_i[j]\\
    &+\sum_{\substack{k_1, .., k_m \geq 0\\ k_1 + .. +k_m =2} }3 a_1\theta[0]\binom{2}{k_1, .., k_m} \prod_{l \in [m]} \theta[l]^{k_l} \xb_i[l]^{k_l} \xb_i[j]\\
    &+\sum_{\substack{k_1, .., k_m \geq 0\\ k_1 + .. +k_m =3} }3 a_1 \binom{3}{k_1, .., k_m} \prod_{l \in [m]} \theta[l]^{k_l} \xb_i[l]^{k_l}  \xb_i[j]\Bigg),
\end{align*}
where $$\binom{n}{k_1, k_2, .., k_m} = \frac{n!}{k_1! k_2! \cdots k_m!}$$ is the multinomial coefficient.
This way, we have a formula for $\yb$ for each entry in the extended $\widetilde{\xb}_i$.

\end{document}